\documentclass[a4paper,UKenglish]{lipics-v2018}

\usepackage{microtype}

\usepackage[ruled,vlined,linesnumbered]{algorithm2e}
\usepackage{amsmath}
\usepackage{todonotes}
\usepackage{hyperref}

\newcommand{\hide}[1]{}
\newcommand{\diam}{{\tt diam}}
\newcommand{\doap}{\mbox{\sc Doap}}
\newcommand{\wdoap}{\mbox{\sc WDoap}}

\nolinenumbers 

\title{Almost optimal algorithms for diameter-optimally augmenting trees}

\titlerunning{Almost optimal algorithms for diameter-optimally augmenting trees}

\author{Davide Bilò}{Department of Humanities and Social Sciences, University of Sassari\\{Via Roma 151, 07100 Sassari (SS), Italy}}{davide.bilo@uniss.it}{https://orcid.org/0000-0003-3169-4300}{}

\authorrunning{D. Bilò}

\Copyright{Davide Bilò}

\subjclass{\ccsdesc[100]{Theory of computation~Graph algorithms analysis}\\\ccsdesc[100]{Theory of computation~Approximation algorithms analysis}}

\keywords{Graph diameter, augmentation problem, trees, time-efficient algorithms.}

\category{}

\relatedversion{}

\supplement{}

\funding{}

\acknowledgements{}

\EventEditors{}
\EventLongTitle{}
\EventShortTitle{}
\EventAcronym{}
\EventYear{}
\EventDate{}
\EventLocation{}
\EventLogo{}
\SeriesVolume{}
\ArticleNo{} 

\begin{document}

\maketitle

\begin{abstract}
We consider the problem of augmenting an $n$-vertex tree with one shortcut in order to minimize the diameter of the resulting graph. The tree is embedded in an unknown space and we have access to an oracle that, when queried on a pair of vertices $u$ and $v$, reports the weight of the shortcut $(u,v)$ in constant time.
Previously, the problem was solved in $O(n^2 \log^3 n)$ time for general weights [Oh and Ahn, ISAAC 2016], in $O(n^2 \log n)$ time for trees embedded in a metric space [Gro{\ss}e et al., {\tt arXiv:1607.05547}],  and in $O(n \log n)$ time for paths embedded in a metric space [Wang, WADS 2017]. Furthermore, a $(1+\varepsilon)$-approximation algorithm running in $O(n+1/\varepsilon^{3})$ has been designed for paths embedded in $\mathbb{R}^d$, for constant values of $d$ [Gro{\ss}e et al., ICALP 2015].

The contribution of this paper is twofold: we address the problem for trees (not only paths) and we also improve upon all known results. More precisely, we design a {\em time-optimal} $O(n^2)$ time algorithm for general weights. Moreover, for trees embedded in a metric space, we design (i) an exact $O(n \log n)$ time algorithm and (ii) a $(1+\varepsilon)$-approximation algorithm that runs in $O\big(n+ \varepsilon^{-1}\log \varepsilon^{-1}\big)$ time.
\end{abstract}

\section{Introduction}

Consider a tree $T=(V(T),E(T))$ of $n$ vertices, with a {\em weight} $\delta(u,v)>0$ associated with each edge $(u,v) \in E(T)$, and let $c:V(T)^2 \rightarrow \mathbb{R}_{\geq 0}$ be an unknown function that assigns a weight to each possible {\em shortcut} $(u,v)$ we could add to $T$. For a given path $P$ of an edge-weighted graph $G$, the {\em length} of $P$ is given by the overall sum of its edge weights. We denote by $d_G(u,v)$ the {\em distance} between $u$ and $v$ in $G$, i.e., the length of a shortest path between $u$ and $v$ in $G$.\footnote{If $u$ and $v$ are in two different connected components of $G$, then $d_G(u,v)=\infty$.} The {\em diameter} of $G$ is the maximum distance between any two vertices in $G$, that is $\max_{u,v \in V(G)}d_G(u,v)$.

In this paper we consider the {\em Diameter-Optimal Augmentation Problem} ($\doap$ for short). More precisely, we are given an edge-weighted tree $T$ and we want to find a shortcut $(u,v)$ whose addition to $T$ minimizes the diameter of the resulting (multi)graph, that we denote by $T+(u,v)$. We assume to have (unlimited access to) an oracle that is able to report the weight of a queried shortcut in $O(1)$ time.

$\doap$ has already been studied before and the best known results are the following:
\begin{itemize}
\item an $O(n^2 \log^3 n)$ time and $O(n)$ space algorithm and a lower bound of $\Omega(n^2)$ on the time complexity of any exact algorithm~\cite{DBLP:conf/isaac/0001A16a};
\item an $O(n^2 \log n)$ time algorithm for trees {\em embedded} in a metric space~\cite{DBLP:journals/corr/GrosseGKSS16};
\item an $O(n \log n)$ time algorithm for paths embedded in a metric space~\cite{DBLP:conf/wads/wang};\footnote{More precisely, $c$ is a metric function and $\delta(u,v)=c(u,v)$, for every $(u,v) \in E(G)$.}
\item a $(1+\varepsilon)$-approximation algorithm that solves the problem in $O(n+1/\varepsilon^{3})$ for paths embedded in the Euclidean (constant) $k$-dimensional space~\cite{DBLP:conf/icalp/GrosseGKSS15}.
\end{itemize} 
In this paper we improve upon (almost) all these results. More precisely:
\begin{itemize}
\item we design an $O(n^2)$ time and space algorithm that solves $\doap$. We observe that the time complexity of our algorithm is optimal;
\item we develop an $O(n \log n)$ time and $O(n)$ space algorithm that solves $\doap$ for trees embedded in a metric space;
\item we provide a $(1+\varepsilon)$-approximation algorithm, running in $O\left(n+\frac{1}{\varepsilon}\log\frac{1}{\varepsilon}\right)$ time and using $O(n+1/\varepsilon)$ space, that solves $\doap$ for trees embedded in a metric space. 
\end{itemize}
Our approaches are similar in spirit to the ones already used in~\cite{DBLP:conf/icalp/GrosseGKSS15,DBLP:journals/corr/GrosseGKSS16,DBLP:conf/wads/wang}, but we need many new key observations and novel algorithmic techniques to extend the results to trees. Our results leave open the problem of solving $\doap$ in $O(n^2)$ time and truly subquadratic space for general instances, and in $o(n \log n)$ time for trees embedded in a metric space. 

\subparagraph*{Other related work.}
The variant of $\doap$ in which we want to minimize the {\em continuous diameter}, i.e., the diameter measured with respect to all the points of a tree (not only its vertices), has been also addressed. Oh and Ahn~\cite{DBLP:conf/isaac/0001A16a} designed an $O(n^2 \log^3 n)$ time and $O(n)$ space algorithm. De Carufel et al.~\cite{DBLP:conf/swat/CarufelMS16} designed an $O(n)$ time algorithm for paths embedded in the Euclidead plane. Subsequently, De Carufel et al.~\cite{DBLP:conf/wads/CarufelGSS17}  extended the results to trees embedded in the Euclidean plane by designing an $O(n \log n)$ time algorithm.
 
Several generalizations of $\doap$ in which the graph (not necessarily a tree) can be augmented with the addition of $k$ edges have also been studied. In the more general setting, the problem is NP-hard~\cite{DBLP:journals/jgt/SchooneBL87}, not approximable within logarithmic factors~\cite{DBLP:journals/tcs/BiloGP12},  and some of its variants -- parameterized w.r.t. the overall cost of added shortcuts and resulting diameter -- are even W$[2]$-hard~\cite{DBLP:journals/algorithmica/FratiGGM15, DBLP:journals/dam/GaoHN13}. Therefore, several approximation algorithms have been developed for all these variations~\cite{DBLP:journals/tcs/BiloGP12, DBLP:journals/algorithmica/ChepoiV02, DBLP:conf/swat/DemaineZ10, DBLP:journals/algorithmica/FratiGGM15, LI1992303}. Finally, upper and lower bounds on the values of the diameters of the augmented graphs have also been investigated in~\cite{DBLP:journals/jgt/AlonGR00, DBLP:journals/jgt/ChungG84, DBLP:journals/jgt/Ishii13}.

\subparagraph*{Our approaches.}
Gro{\ss}e et al.~\cite{DBLP:conf/icalp/GrosseGKSS15} were the first to attack $\doap$ for paths embedded in a metric space. They gave an $O(n \log n)$ time algorithm for the corresponding {\em search version} of the problem: 
\begin{quote}
\textit{For a given value $\lambda > 0$, either compute a shortcut whose addition to the path induces a graph of diameter at most $\lambda$, or return $\perp$ if such a shortcut does not exist.}
\end{quote}
Then, by implementing their algorithm also in a parallel fashion and applying Megiddo's parametric-search paradigm~\cite{DBLP:journals/jacm/Megiddo83}, they solved $\doap$ for paths embedded in a metric space in $O(n \log^3 n)$ time. Lately, Wang~\cite{DBLP:conf/wads/wang} improved upon this result in two ways. First, he solved the search version of the problem in linear time. Second, he developed an ad-hoc algorithm that, using the algorithm for the search version of the problem black-box together with sorted-matrix searching techniques and range-minima data structure, is able to: (i) reduce the size of the solution-search-space from $\binom{n}{2}$ to $n$ in $O(n \log n)$ and (ii) evaluate the quality of all the leftover solutions in $O(n)$ time.

Our approach for $\doap$ instances embedded in a metric space is close in spirit to the approach used by Wang. In fact, we develop an algorithm that solves the search version of $\doap$ in linear time and we use such an algorithm black-box to solve $\doap$ in $O(n \log n)$ time and linear space by first reducing the size of the solution-search-space from $\binom{n}{2}$ to $n$ and then by evaluating the quality of the leftover solutions in $O(n \log n)$ time. However, differently from Wang's approach, we use  Hershberger data structure for computing the {\em upper envelope} of a set of linear functions~\cite{DBLP:journals/ipl/Hershberger89} rather than a range-minima data structure. Furthermore, there are several issues we have to deal with due to the much more complex topology of  trees. We solve some of these issues using a lemma proved in~\cite{DBLP:journals/corr/GrosseGKSS16} about the existence of an optimal shortcut whose endvertices both belong to a diametral path of the tree. This allows us to reduce our $\doap$ instance to a node-weighted path instance of a similar problem, that we call $\wdoap$, in which the distance between two vertices is measured by adding the weights of the two considered vertices to the length of a shortest path between them, and the diameter is defined accordingly. However, it is not possible to use the algorithms presented in~\cite{DBLP:conf/icalp/GrosseGKSS15, DBLP:conf/wads/wang} black-box to solve $\wdoap$. Therefore we need to design an ad-hoc algorithm whose correctness strongly relies on the structural properties of diametral paths and properties satisfied by node weights. Furthermore, most of the easy observations that can be done for paths become non-trivial lemmas that need formal proofs for trees.

Our time-optimal algorithm that solves $\doap$ for instances with general weights is based on the following important observations. We reduce, in $O(n^2)$ time, a $\doap$  instance to another $\doap$ instance in which the function $c$ is {\em graph-metric}, i.e., $c$ is an almost metric function that satisfies a weaker version of the triangle inequality. Since our $O(n \log n)$ time algorithm for $\doap$ instances embedded in a metric space also works for graph-metric spaces, we can use this algorithm black-box to solve the reduced $\doap$ instance in $O(n \log n)$ time, thus solving the original $\doap$ instance in $O(n^2)$ time.

Finally, the $(1+\varepsilon)$-approximation algorithm for trees embedded in a metric space is obtained by  proving that the diameter of the tree is at most three times the diameter, say $D^*$, of an optimal solution. This allows us to partition the vertices along a diametral path into $O(1/\varepsilon)$ sets such that the distance between any two vertices of the same set is at most $O(\varepsilon D^*)$. We choose a suitable {\em representative} vertex for each of the $O(1/\varepsilon)$ sets and use our $O(n \log n)$ time algorithm to find an optimal shortcut in the corresponding $\wdoap$ instance restricted to the set of representative vertices. Since the representative vertices are $O(1/\varepsilon)$, the optimal shortcut in the restricted $\wdoap$ instance can be found in $O(\varepsilon^{-1}\log \varepsilon^{-1})$ time. Furthermore, because of the choice of the representative vertex, we can show that the shortcut returned is a $(1+\varepsilon)$-approximate solution for the (unrestricted) $\wdoap$ instance of our problem, i.e., a $(1+\varepsilon)$-approximate solution for our original $\doap$ instance.


\subparagraph*{Paper organization.}
The paper is organized as follows: in Section~\ref{section:preliminaries} we present some preliminary results among which the reduction from general instances to (graph)-metric instances; in Section~\ref{section:reduction_to_paths} we describe the reduction from $\doap$ to $\wdoap$ together with further simplifications; in Section~\ref{section:search_problem} we design an algorithm that solves a search version of $\wdoap$ in linear time; in Section~\ref{section:metric_trees} we develop an algorithm that solves $\doap$ for trees embedded in a metric space; in Section~\ref{section:general_instances} we describe the $O(n^2)$ time algorithm that solves $\doap$; in Section~\ref{section:approximation_algorithm} we design the linear time approximation algorithm that finds a $(1+\varepsilon)$-approximate solution for instances of $\doap$ embedded in a metric space. 

\section{Preliminaries}\label{section:preliminaries}

To simplify the notation, we drop the subscript from $d_T(\cdot,\cdot)$ whenever $T$ is clear from the contest and we denote $d_{T+(u,v)}(\cdot,\cdot)$ by $d_{u,v}(\cdot,\cdot)$. The diameter of a graph $G$ is denoted by $\diam(G)$. A {\em diametral path} of $G$ is a shortest path in $G$ of length equal to $\diam(G)$. We say that $c$ is a {\em graph-metric w.r.t. $G$}, or simply a {\em graph-metric} when $G$ is clear from the contest, if, for every three distinct vertices $u, v$, and $z$ of $G$, we have that
\[
c(u,v) \leq c(u,z)+d(z,v). \tag{{\em graph-triangle inequality}}
\]
We observe that a metric cost function is also graph-metric, but the opposite does not hold in general (see Figure \ref{fig:graph_metric_function_example}). 
\begin{figure}[ht]
	\centering
	\includegraphics[scale=0.8]{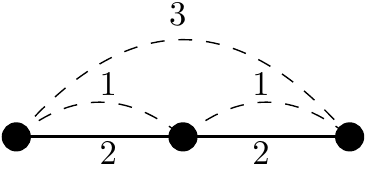}
	\caption{An example of a graph-metric function. The graph (a path in this specific example) is given by the two solid edges of weight $2$ each. The shortcuts are dashed. The example shows a graph-metric function that does not satisfy the triangle inequality.}
    \label{fig:graph_metric_function_example}
\end{figure}
The {\em graph-metric closure} induced by $c$ is a function $\bar c$ such that, for every two vertices $u$ and $v$ of $G$, $\bar c(u,v)=\min\big\{d_G(u,u')+c(u',v')+d_G(v',v) \mid u', v' \in V(G)\big\}$.
The following lemma shows that we can restrict $\doap$ to input instances where $c$ is graph-metric. We observe that the reduction holds for any graph and not only for trees.
\begin{lemma}\label{lemma:reduction_to_tree_metric_instances}
Solving the instance $\langle G, \delta, c \rangle$ of $\doap$ is equivalent to solving the instance $\langle G, \delta, \bar c \rangle$ of $\doap$, where $\bar c$ is the graph-metric closure induced by $c$.
\end{lemma}
\begin{proof}
We prove the claim by showing that for every two vertices $u$ and $v$ in $G$, there exists two vertices $u'$ and $v'$ in $G$ such that the diameter of $G+(u',v')$ measured w.r.t. $\bar c$ (resp., $c$) is at most the diameter of $G+(u,v)$ measured w.r.t. $c$ (resp. $\bar c$). 

Let $u, v$ be any two vertices of $G$. Since $\bar c(u,v) \leq c(u,v)$, the diameter of $G+(u,v)$ measured w.r.t. $\bar c$ is at most the diameter of the same graph measured w.r.t. $c$ (i.e., $u'=u$ and $v'=v$). To prove the converse for suitable vertices $u'$ and $v'$, let $u'$ and $v'$ be two vertices such that $\bar c(u,v)=d(u,u')+c(u',v')+d(v',v)$.
Every path in $G+(u,v)$ passing through $(u,v)$ can be replaced in $G+(u', v')$ by the same path where the edge $(u,v)$ is bypassed with the detour passing through $(u',v')$. Since the cost of this detour is exactly $\bar c(u,v)$, the diameter of $G+(u',v')$ measured w.r.t. $c$ is less than or equal to the diameter of $G+(u,v)$ measured w.r.t. $\bar c$. This completes the proof.
\end{proof}
Next lemma shows the existence of an optimal shortcut whose endvertices are both on a diametral path of $T$ for the case in which $c$ is a graph-metric. 
\begin{lemma}\label{lemma:equivalence_to_metric_instance}
Let $\langle T, \delta, c\rangle$ be an instance of \doap, where $c$ is a graph-metric, and let $P=(v_1,\dots,v_N)$ be a diametral path of $T$. There always exists an optimal shortcut $(u^*,v^*)$ such that $u^*,v^* \in V(P)$.
\end{lemma}
\begin{proof}
First of all, observe that $\diam(T) = d(v_1,v_N)$. Let $u, v \in V(T)$ be two vertices such that $\diam\big(T+(u,v)\big)$ is minimum. Let $v_i$ be the last vertex of $P$ encountered during the traversal of the path from $v_1$ to $u$ in $T$, and let $v_j$ be the last vertex of $P$ encountered during the traversal of the path from $v_1$ to $v$ in $T$. W.l.o.g., we assume that $i \leq j$. We prove the claim by showing that $\diam\big(T+(v_i,v_j)\big)\leq \diam\big(T+(u,v)\big)$.

\noindent We assume that $d_{u,v}(v_1,v_N) < d(v_1,v_N)$ as otherwise the claim would trivially hold since $d_{v_i,v_j}(u,v) \leq d(v_1,v_N)$. This implies that 
\begin{equation}\label{eq:1}
d_{u,v}(v_1,v_N)=d(v_1,v_i)+d(v_i,u)+c(u,v)+d(v,v_j)+d(v_j,v_N)
\end{equation}
as well as that $i \neq j$. Therefore, $i < j$.

We prove the claim by showing that $d_{v_i,v_j}(a,b)\leq \diam\big(T+(u,v)\big)$, for every two vertices $a$ and $b$ of $T$. Let $a$ and $b$ be any two fixed vertices of $T$. Let $v_k$ be the last vertex of $P$ that is encountered during a traversal of the path in $T$ from $v_1$ to $a$. Similarly, let $v_h$ be the last vertex of $P$ that is encountered during the traversal of the path in $T$ from $v_1$ to $b$. W.l.o.g., we assume that $k \leq h$. We can rule out the case in which $d_{u,v}(a,b)=d(a,b)$ since this would imply $d_{v_i,v_j}(a,b) \leq d(a,b) =d_{u,v}(a,b)\leq \diam\big(T+(u,v)\big)$ and thus the claim. Therefore, in the following we assume that $d_{u,v}(a,b) < d(a,b)$. As a consequence, we have that $k < h$ since $k=h$ implies $d_{v_i,v_j}(a,b)=d(a,b)$. Moreover, $d_{u,v}(a,b) = d(a,u)+c(u,v)+d(v,b)$.
We break the proof into the following three cases: 
\begin{itemize}
\item $k \neq i$ and $h \neq j$;
\item both $k \leq i$ and $h \geq j$ hold;
\item either $k = i$ and $h < j$ holds or $k > i$ and $h = j$ holds.
\end{itemize}

\noindent We consider the first case in which $k \neq i$ and $h \neq j$ (see Figure \ref{fig:lemma_optimal_shortcut}, Case 1). Since $c(v_i,v_j) \leq d(v_i,u)+c(u,v)+d(v,v_j)$, we have that
\begin{align*}
d_{v_i,v_j}(a,b)	& \leq d(a,v_i)+c(v_i,v_j)+d(v_j,b)\\
						& \leq d(a,v_i)+d(v_i,u)+c(u,v)+d(v,v_j)+d(v_j,b) \\
						& \leq d(a,u)+c(u,v)+d(v,b)\\
						& = d_{u,v}(a,b) \leq \diam\big(T+(u,v)\big).
\end{align*}
\begin{figure}[ht]
	\centering
	\includegraphics[scale=0.8]{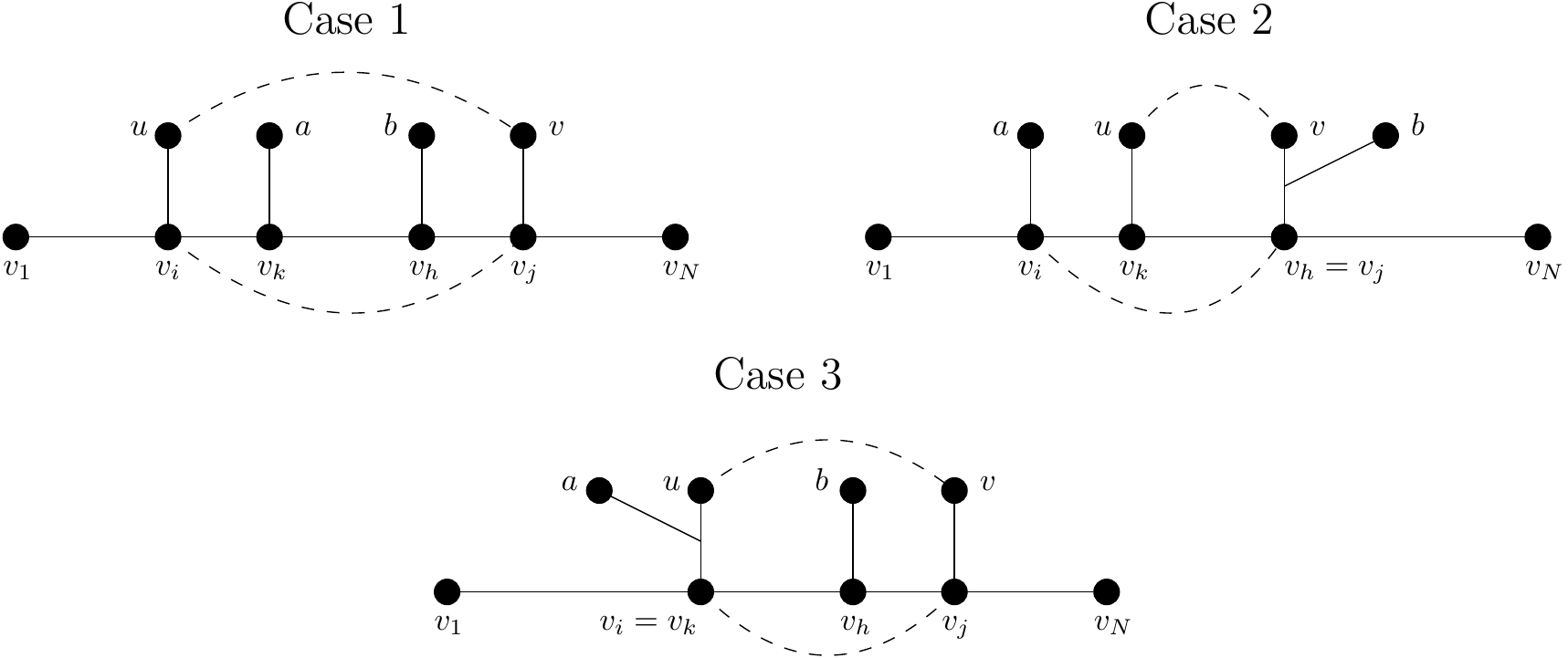}
	\caption{The three-case analysis in the proof of Lemma \ref{lemma:equivalence_to_metric_instance}.}
    \label{fig:lemma_optimal_shortcut}
\end{figure}

We consider the second case in which both $k \leq i$ and $h \geq j$ (see Figure \ref{fig:lemma_optimal_shortcut}, Case 2). First of all, observe that $d(a,v_i) \leq d(v_1,v_i)$ as well as $d(v_j,b) \leq d(v_j,v_N)$. Therefore, since $c(v_i,v_j) \leq d(v_i,u)+c(u,v)+d(v,v_j)$, using (\ref{eq:1}) in the last equality that follows, we obtain
\begin{align*}
d_{v_i,v_j}(a,b) 	& \leq d(a,v_i)+c(v_i,v_j)+d(v_j,b)\\
						& \leq d(v_1,v_i) + d(v_i,u) + c(u,v) + d(v,v_j) + d(v_j,v_N)\\
						& = d_{u,v}(v_1,v_N) \leq \diam\big(T+(u,v)\big).
\end{align*}
We consider the last case in which either $k = i$ and $h < j$ holds or $k > i$ and $h = j$ holds (see Figure \ref{fig:lemma_optimal_shortcut}, Case 3). W.l.o.g., we assume that $k = i$ and $h < j$ as the proof for the (symmetric) case in which $k > i$ and $h = j$ is similar. Since $d(a,v_i)\leq d(v_1,v_i)$ we have that
\begin{align*}
d_{v_i,v_j}(a,b)	& = d(a,v_i)+\min\big\{d(v_i,b),c(v_i,v_j)+d(v_j,b)\big\}\\
						& \leq d(v_1,v_i) + \min\big\{d(v_i,b),d(v_i,u)+c(u,v)+d(v,v_j)+d(v_j,b)\big\}\\
						& = d_{u,v}(v_1,b) \leq \diam\big(T+(u,v)\big).
\end{align*}
This completes the proof.
\end{proof}
\section{Reduction from trees to node-weighted paths}\label{section:reduction_to_paths}

In this section we show that a $\doap$ instance embedded in a graph-metric space can be reduced in linear time to a node-weighted instance of a similar problem. The {\em Node-Weighted-Diameter-Optimal Augmentation Problem} ($\wdoap$ for short) is defined as follows:
\begin{description}
	\item[Input:] A path $P=(v_1,\dots,v_N)$, with a weight $\delta(v_i,v_{i+1})>0$ associated with each edge $(v_i,v_{i+1})$ of $P$, a weight $w(v_i)$ associated with each vertex $v_i$ such that $0 \leq w(v_i) \leq \min\{d(v_1,v_i),d(v_i,v_N)\}$, and an oracle that is able to report the weight $c(v_i,v_j)$ of a queried shortcut in $O(1)$ time, where $c$ is a graph-metric;
		
	\item[Output:] Two indices $i^*$ and $j^*$, with $1 \leq i^* < j^* \leq N$, that minimize the function
	$$
		D(i,j):=\max_{1\leq k < h \leq N}\big\{w(v_k) + d_{v_i,v_j}(v_k,v_h) + w(v_h)\big\}.
	$$
\end{description}
We observe that $w(v_1)=w(v_N)=0$.
Let $\langle T, \delta, c\rangle$ be a $\doap$ instance, where $c$ is a graph-metric. Let $P=(v_1,\dots,v_N)$ be a diametral path of $T$, $T_i$ the tree containing $v_i$ in the forest obtained by removing the edges of $P$ from $T$, and $w(v_i):=\max_{v \in V(T_i)}d(v_i,v)$. We say that $\langle P, \delta, w, c\rangle$ is the $\wdoap$ instance induced by $\langle T, \delta, c\rangle$ and $P$. The following lemma holds.

\begin{lemma}\label{lemma:discard_diameter_appendedrees}
For every $1 \leq k \leq N$, $D(i,j) \geq \diam(T_k)$.
\end{lemma}
\begin{proof}
Since $P$ is a diametral path of $T$, $d(v_1,v_k),d(v_k,v_N) \geq \diam(T_k)$. Furthermore, if $a$ and $b$ are the endvertices of a diametral path in $T_k$, then $d(a,v_k)+d(b,v_k)\geq \diam(T_k)$, i.e., $w(v_k) \geq \max\big\{d(a,v_k),d(b,v_k)\big\} \geq \diam(T_k)/2$. Observe that either the shortest path between $v_k$ and $v_1$ or the one between $v_k$ to $v_N$ does not pass through $(i,j)$ in $P+(i,j)$. W.l.o.g., assume that the shortest path between $v_k$ and $v_1$ in $P+(i,j)$ does not pass through $(i,j)$. As a consequence, the shortest path between $v$ and $v_1$ in $P+(i,j)$ does not pass through $(i,j)$. Furthermore, the cost of this path is at least $d(v_1,v_k)+d(v_k,v) \geq 2d(v_k,v) \geq \diam(T_k)$. The claim follows.
\end{proof}
The proof of the following lemma makes use of Lemma \ref{lemma:discard_diameter_appendedrees}.
\begin{lemma}\label{lemma:Dij}
The $\wdoap$ instance $\langle P, \delta, w, c\rangle$  induced by $\langle T, \delta, c\rangle$ and $P$ can be computed in $O(n)$ time. Moreover,  $\diam\big(T+(v_i,v_j)\big)=D(i,j)$, for every $1 \leq i < j \leq N$.
\end{lemma}
\begin{proof}
Concerning the time required to compute the $\wdoap$ instance $\langle P, \delta, w, c\rangle$, we can easily compute a diametral path $P=(v_1,\dots,v_N)$ and all the values $w(v_i)$, for every $i=1,\dots,N$, in $O(n)$ time. Therefore, the reduction takes $O(n)$ time.

Let $e=(v_i,v_j)$. We prove that $D(i,j) \leq \diam(T+e)$. Let $v_k$ and $v_h$, with $k < h$, be any two vertices of $P$. Let $u$ be a vertex of $T_k$ such that $d(u,v_k)=w(v_k)$. Similarly, let $v$ be a vertex of $T_h$ such that $d(v,v_h) = w(v_h)$. We have that $w(v_k)+d_{P+e}(v_k,v_h) + w(v_h) = d(u,v_k) + d_{T+e}(v_k, v_h) +  d(v_h,v) = d_{T+e}(u,v) \leq \diam(T + e)$. Therefore, $D(i,j) \leq \diam\big(T+e)$.

Now, we prove that $\diam(T+e) \leq D(i,j)$.
Let $u,v$ be any two vertices of $T$ such that $u \in V(T_k)$ and $v \in V(T_h)$, with $k \leq h$. Using Lemma~\ref{lemma:discard_diameter_appendedrees}, if $k=h$, then $\diam(T+e) \leq D(i,j)$. If $k < h$, then
\begin{align*}
d_{T+e}(u,v)	& \leq d(u, v_k) + d_{T+e}(v_k, v_h) + d(v_h, v)  \leq w(v_k) + d_{T+e}(v_k,v_h) + w(v_h)\\
				& = w(v_k)+d_{P+e}(v_k,v_h)+w(v_h) \leq D(i,j).
\end{align*}
This completes the proof.
\end{proof}

\subsection{Further simplifications}

In the rest of the paper, we show how to solve $\wdoap$ in $O(N\log N)$ time and linear space.
To avoid heavy notation, from now on we denote a vertex $v_i$ by using its associated index $i$. All the lemmas contained in this subsection are non-trivial generalizations of observations made in~\cite{DBLP:conf/icalp/GrosseGKSS15} for paths. We start proving a useful lemma.
\begin{lemma}\label{lemma:diameter_node_weighted_path}
Let $i,j$ be two indices such that $1\leq i < j \leq N$. Let $I=\{1\}\cup \{k \mid i < k \leq N\}$ and let $J=\{N\}\cup \{h \mid 1 \leq h < j\}$. We have that
$$
D(i,j) = \max_{k \in I, h \in J, k < h} \big\{ w(k) + d_{i,j}(k,h) + w(h) \big\}.
$$
\end{lemma}
\begin{proof}
Let $\alpha=\max_{k \in I, h \in J, k < h} \big\{ w(k) + d_{i,j}(k,h) + w(r) \big\}$. Clearly, $D(i,j) \geq \alpha$. Now we show that $D(i,j) \leq \alpha$. Let $1 \leq k^* < h^* \leq N$ be such that $D(i,j)=w(k^*)+d_{i,j}(k^*,h^*)+w(h^*)$.
For any $k$, with $2 \leq k \leq i$, and for any $h$, with $k < h \leq N$, $w(k) + d_{i,j}(k,h) + w(h) \leq d(1,k) + d_{i,j}(k,h) + w(h) = d_{i,j}(1,h) + w(h) \leq \alpha$. Therefore, either $k^*=1$ or $k^* > i$, i.e., $k^* \in I$.
Similarly, for any $h$, with $j \leq h \leq N-1$, and for any $k$, with $1 \leq k < h$, $w(k) + d_{i,j}(k,h) + w(h) \leq w(k) + d_{i,j}(k,h) + d(h,N) = w(k) + d_{i,j}(k,N)+w(N)\leq \alpha$. Therefore, either $h^*=N$ or $h^* < j$, i.e., $h^* \in J$. The claim follows.
\end{proof}
As we will see in a short, Lemma~\ref{lemma:diameter_node_weighted_path} allows us to decompose the function $D(i, j)$ into four monotone parts. First of all, for every $i=1,\dots,N$, we define 
$$
\omega(i) := \max\big\{w(j)-d(i,j) \mid 1 \leq j \leq N\big\}.
$$
Observe that, for every $1 \leq i \leq j \leq N$,
\[
\omega(i) \leq \omega(j)+d(i,j). \tag{{\em node-triangle inequality}}
\]
 Furthermore, $\omega(i) \geq w(i)$, for every $1 \leq i \leq N$, which implies $\omega(1)=\omega(N)=0$. The following lemma establishes the time complexity needed to compute all the values $\omega(i)$.
\begin{lemma}\label{lemma:computation_omega}
All the values $\omega(i)$, with $1 \leq i \leq N$, can be computed in $O(N)$ time.
\end{lemma}
\begin{proof}
Observe that, for every $1 \leq i \leq N$, all the values
$$
\omega'(i)=\max\big\{w(j)-d(i,j) \mid i \leq j \leq N\big\}
$$
can be computed in $O(N)$ time by scanning all the vertices of $P$ from $N$ downto 1. Indeed, $\omega'(N)=w(N)$ and, for every $1 \leq i < N$, $\omega'(i)=\max\{w(i),\omega'(i+1)-\delta(i,i+1)\}$.
Observe that, for every $1 \leq i \leq N$, 
$$
\omega(i) = \max\{\omega'(j)-d(j,i) \mid 1 \leq j \leq i\}.
$$
As a consequence, all the values $\omega(i)$ can be computed from the values $\omega'(\cdot)$ by scanning all the vertices of $P$ from $1$ to $N$. Indeed, $\omega(1)=\omega'(1)$ and, for every $1 < i \leq N$, $\omega(i)=\max\{\omega'(i),\omega'(i-1)-\delta(i-1,i)\}$.
\end{proof}
For the rest of this section, unless stated otherwise, $i$ and $j$ are such that $1 \leq i < j \leq N$. The four functions used to decompose $D(i, j)$ are the following (see also Figure \ref{fig:functions})
\begin{figure}[ t]
	\centering
	\includegraphics[scale=0.8]{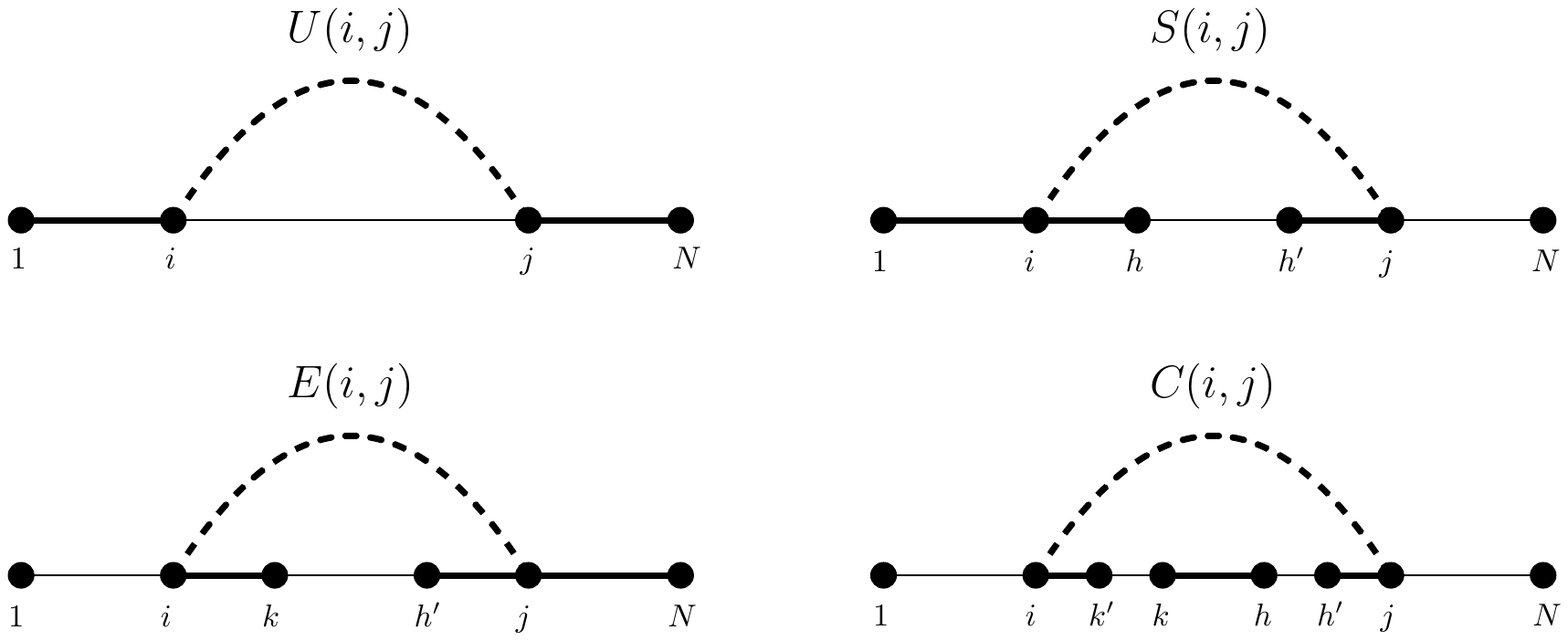}
	\caption{The four functions used to decompose $D(i,j)$. Node weights are omitted and shortest paths are highlighted in bold. $U(i,j)=d(1,i)+c(i,j)+d(j,N)$. $S(i,j)$ is the maximum among all the (node-weighed) distances between $1$ and all the vertices of the cycle. In our example the distance from $1$ to $h$ is $d(1,h)+\omega(h)$, while the distance from $1$ to $h'$ is $d(1,i)+c(i,j)+d(j,h')+\omega(h')$. $E(i,j)$ is the maximum among all the distances between $N$ and all the vertices of the cycle. In our example the distance from $N$ to $k$ is $d(k,N)+\omega(k)$, while the distance from $N$ to $k'$ is $d(j,N)+c(i,j)+d(i,k')+\omega(k')$. Finally, $C(i,j)$ is the maximum among all the distances between pair of vertices in the cycle. In our example the distance from $k$ to $h$ is $\omega(k)+d(k,h)+\omega(h)$, while the distance from $k'$ to $h'$ is $\omega(k')+d(i,k)+c(i,j)+d(j,h')+\omega(h')$.}
    \label{fig:functions}
\end{figure}

$$
U(i, j) := d(1, i) + c(i, j) + d(j, N);
$$
$$
S(i, j) :=  \max_{i \leq h < j} \Big(\omega(h) +\min\big\{d(1, h), d(1, i) + c(i, j) + d(h, j)\big\}\Big);
$$
$$
E(i, j) := \max_{i < k \leq j} \Big(\omega(k) +\min\big\{d(k, N), d(j, N) + c(i, j) + d(i, k)\big\}\Big);
$$
$$
C(i, j) := \max_{i < k < h < j} \Big(\omega(k) + \min\big\{d(k, h), d(i, k) + c(i, j) + d(h, j)\big\} + \omega(h)\Big).
$$
Using both the graph-triangle inequality and the node-triangle inequality, we can observe that all the four functions are monotonic. More precisely:
\begin{itemize}
\item $U(i,j+1) \leq U(i, j) \leq U(i+1, j)$;
\item $S(i-1, j) \leq S(i, j) \leq S(i, j+1)$;
\item $E(i,j+1) \leq E(i, j) \leq E(i-1, j)$;
\item $C(i+1, j) \leq C(i, j) \leq C(i, j + 1)$.
\end{itemize}
Moreover, we can prove the following lemma.
\begin{lemma}\label{lemma:Dij_splitted}
$D(i, j) = \max\big\{U(i, j), S(i, j), E(i, j), C(i, j)\big\}$.
\end{lemma}
\begin{proof}
Let $k$ and $h$, with $1 \leq k < h \leq N$, two indices such that $D(i,j)=w(k)+d_{i,j}(k,h)+w(h)$. Using Lemma~\ref{lemma:diameter_node_weighted_path}, we have that either $k=1$ or $i < k$ as well as either $h = N$ or $h < j$.
Clearly, $D(i,j) \leq U(i,j)$ when $k=1$ and $h = N$, $D(i,j) \leq S(i,j)$ when $k=1$ and $h < j$, $D(i,j) \leq E(i,j)$ when $k > i$ and $h = N$, and $D(i,j) \leq C(i,j)$ when $i < k < h < j$. Therefore, it remains to prove that $D(i,j) \leq \max\big\{U(i, j), S(i, j), E(i, j), C(i, j)\big\}$ when $j > k$ or when $h < i$. We prove the claim for the case $j > k$ as the proof for the case $h < j$ is similar. Using the node-triangle inequality, we have that
$w(k)+d_{i,j}(k,h)+w(h)=w(k)+d(k,h)+w(h)\leq w(j)+d(j,k)+d(k,h)+d(h,N) \leq w(j)+d(j,N) \leq \omega(i)+ d(j,N) \leq E(i,j)$.

It remains to show that $D(i,j) \geq U(i,j), S(i,j), E(i,j), C(i,j)$. Let $1 \leq k < h \leq N$ be two vertices such that the value $\omega(h) + d_{i,j}(k,h) + \omega(k)$ is maximized. Let $k'$, with $1 \leq k' \leq N$, be such that $\omega(k) = w(k') - d_{P}(k, k')$. Using the triangle inequality, we have that $d_{i,j}(k,h) \leq d_{i,j}(k', h) + d(k, k')$, from which we derive that $d_{i,j}(k,h) - d(k, k') \leq d_{i,j}(k', h)$. Let $h'$, with $1 \leq h' \leq N$, be such that $\omega(h) = w(h') - d_{P}(h, h')$. Using the triangle inequality, we have that $d_{i,j}(k',h) \leq d_{i,j}(k', h') + d(h, h')$, from which we derive that $d_{i,j}(k',h) - d(h, h') \leq d_{i,j}(k', h')$. Therefore, by definition of $U(i,j), S(i,j), E(i,j)$, and $C(i,k)$, using Lemma~\ref{lemma:discard_diameter_appendedrees} and Lemma~\ref{lemma:Dij} in the last inequality of the following chain, we have that
\begin{align*}
\max\big\{U(i,j),S(i,j),E(i,j),C(i,j)\big\} & \leq \omega(k) + d_{i,j} (k,h) + \omega(h) \\
											& = w(k') - d(k, k') + d_{i,j}(k,h) + \omega(h)\\
 											& \leq w(k') + d_{i,j}(k', h) + \omega(h) \\
 											& = w(k') + d_{i,j}(k', h) -d(h, h') + w(h') \\
											& \leq w(k') + d_{i,j}(k, h) + w(h') \leq D(i,j).
\end{align*}
This completes the proof.
\end{proof}
The following lemma allows us to efficiently compute the values $U(i,j), S(i,j)$, and $E(i,j)$.
\begin{lemma}\label{lemma:time_complexity_real_functions}
After a $O(N)$-time precomputation phase, for every $1 \leq i < j \leq N$, $U(i,j)$ can be computed in $O(1)$ time, while both $S(i,j)$ and $E(i,j)$ can be computed in $O(\log N)$ time.
\end{lemma}
\begin{proof}
We precompute all the distances $d(1,k)$ and $d(k,N)$, for every $1 \leq k \leq N$, in $O(N)$ time. This allows us to compute $U(i,j)$ in $O(1)$ time.

In the following we prove the claim only for $S(i,j)$, as the proof for $E(i,j)$ is similar. Let $i \leq \bar h \leq j$ be the maximum index such that $d(1,\bar h) \leq d(1,i)+c(i,j)+d(j,\bar h)$. Observe that $\bar h$ can be computed in $O(\log N)$ time using a binary search. We show that once the value of $\bar h$ is known, the value $S(i,j)$ can be computed in $O(1)$ time.
Using the node-triangle inequality, for every $i \leq h \leq \bar h$, we have that 
\begin{align*}
d_{i,j}(1,h)+\omega(h)	& = d(1,h) + \omega(h) \leq d(1,h) + d(h,\bar h) + \omega(\bar h) \\
									& = d(1,\bar h) + \omega(\bar h)  = d_{i,j}(1,\bar h) + \omega(\bar h).
\end{align*}
Using the node-triangle inequality, for every $\bar h < h \leq j$, we have that 
\begin{align*}
d_{i,j}(1,h)+\omega(h) 	& = d(1,i) + c(i,j) + d(h,j) + \omega(h) \\
						& \leq d(1,i) + c(i,j) + d(h,j) + d(\bar h+1,h) + \omega(\bar h+1)\\
						& = d(1,i) + c(i,j) + d(\bar h+1,j) + \omega(\bar h+1) = d_{i,j}(1,\bar h+1) + \omega(\bar h+1).
\end{align*}
Therefore, once $\bar h$ is known,
$$
S(i,j) = \max\big\{d(1,\bar h) + \omega(\bar h), d(1,i)+c(i,j)+d(\bar h+1,j) + \omega(\bar h+1)\big\}.
$$
can be computed in constant time.
\end{proof}
\section{The linear time algorithm for the search version of \wdoap}\label{section:search_problem}

In this section we design an $O(N)$ time algorithm (see Algorithm~\ref{algorithm:search_problem}) for the following search version of $\wdoap$:
\begin{quote}
\textit{For a given $\wdoap$ instance $\langle P, \delta, \omega, c\rangle$, where $c$ is a graph-metric and $\omega$ satisfies the node-triangle inequality, and a real value $\lambda > 0$, either find two indices $1 \leq i < j \leq N$ such that $D(i,j) \leq \lambda$, or return $\perp$ if such two indices do not exist.}
\end{quote}
In the following we assume that $d(1,N) > \lambda$, as otherwise $D(i,j) \leq \lambda$ for any two indices $i$ and $j$.
For the rest of this section, unless stated otherwise, $i$ and $j$ are two fixed indices such that $1 \leq i < j \leq N$. Let $i < \mu_i \leq N$ be the minimum index, or $N+1$ if such an index does not exists, such that $U(i,\mu_i) \leq \lambda$. Our algorithm computes the index $\mu_i$, for every $1 \leq i < N$. As $U(i,j) \geq U(i,j+1)$ for every $i < j < N$, the following lemma holds.
\begin{lemma}\label{lemma:u_function_evaluation}
$U(i,j) \leq \lambda$ iff $\mu_i \leq j$ (see also Figure \ref{fig:search_problem_example}).
\end{lemma}
\begin{figure}[ht]
	\centering
	\includegraphics[scale=0.8]{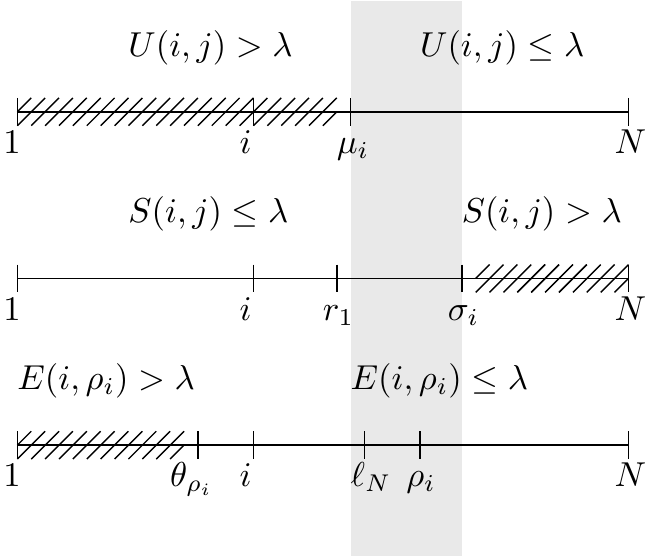}
	\caption{An example showing the properties satisfied by the functions $U(i,j), S(i,j),$ and $E(i,\rho_i)$. The example shows a case in which $\rho_i$ is defined. We observe that the search version of $\wdoap$ admits a feasible solution consisting of a shortcut adjacent to $i$ iff there exists an index $j$ belonging to the shaded area such that $C(i,j) \leq \lambda$. Furthermore, among all the possible choices, $\rho_i$ is the one that minimizes the value $C(i,j)$.}
    \label{fig:search_problem_example}
\end{figure}
Moreover, as $U(i,j) \leq U(i+1,j)$, we have that $\mu_i \leq \mu_{i+1}$. Therefore, our algorithm can compute all the indices $\mu_i$ in $O(N)$ time by scanning all the vertices of $P$ from $1$ to $N$.

\begin{algorithm}[ht]
	\SetKwInOut{Input}{Input}\SetKwInOut{Output}{Output}
	\Input{An instance $\langle P, \delta, \omega, c \rangle$ of $\wdoap$, and an integer $\lambda > 0$;}
	\Output{Two indices $i$ and $j$, with $1 \leq i < j \leq N$, such that $D(i,j) \leq \lambda$, or $\perp$ if such two indices do not exist.}
	{$r_{N} = N$; $\ell_N = 1$; $\mu_0=1$; $\sigma_{0} = N$; $\theta_N = 1$; $\rho_{0} = 1$}\;
	\For{$i$ from $N-1$ downto $1$}
	{
		$r_i = r_{i+1}$\;
		\lWhile{$r_i > i$ and $\omega(i)+d(i,r_i)+\omega(r_i) > \lambda$}{$r_i = r_i - 1$}
	}
	\lWhile{$\ell_N < N$ and $\omega(\ell_N)+d(\ell_N,N) > \lambda$}{$\ell_N = \ell_N + 1$}
	\For{$i$ from $1$ to $N$}
	{
		$\mu_i = \max\{i+1,\mu_{i-1}\}$\; 
		\lWhile{$i < N$ and $U(i,\mu_i) > \lambda$}{$\mu_i = \mu_i + 1$}
	}
	
	\For{$i$ from $1$ to $N$}
	{
		$\sigma_i = \sigma_{i-1}$\;
		\lWhile{$\sigma_i > i$ and $\bar S(i, \sigma_i) > \lambda$}{$\sigma_i = \sigma_i - 1$}
	}

	\For{$j$ from $N$ downto $1$}
	{
		$\theta_j = \theta_{j+1}$\;
		\lWhile{$\theta_j < j$ and $\bar E(\theta_j,j) > \lambda$}{$\theta_j = \theta_j + 1$}
	}
	\For{$i$ from $1$ to $N$}
	{
		{$\rho_i = \rho_{i-1}$\;}
		\lWhile{$\rho_i \leq N$ and $\theta_{\rho_i} > i$}{$\rho_i = \rho_i + 1$}
		\lIf{$\rho_i \leq N$ and $\mu_i \leq \rho_i \leq \sigma_i$}
		{$B[i] = 1$; {\bf else} $B[i] = 0$}
	}
	
	{$\Delta_{\min} = \infty$\;}	
	\For{$i$ from $1$ to $N$}
	{
		\lIf{$r_i < N$ and $\lambda-\omega(i)+d(i,r_i+1)-\omega(r_i+1) < \Delta_{\min}$}{$\Delta_{\min}=\lambda-\omega(i)+d(i,r_i+1)-\omega(r_i+1)$}
	}	
	
	\For{$i$ from $1$ to $N$}
	{
		\lIf{$B[i] = 1$ and $d(i,\rho_i)+c(i, \rho_i) < \Delta_{\min}$}{\Return $(i, \rho_i)$}
	}	
	
	\Return $\perp$\;
	\caption{The algorithm for the search version of $\wdoap$.}\label{algorithm:search_problem}
\end{algorithm}

We introduce some new notation useful to describe our algorithm. We define $r_i$ as the maximum index such that  $i < r_i \leq N$ and $\omega(i) + d(i,r_i)+\omega(r_i) \leq \lambda$. If such an index does not exist, we set $r_i=i$. Similarly, we define $\ell_N$ as the minimum index such that $1 \leq \ell_N < N$ and $\omega(\ell_N) + d(\ell_N,N) \leq \lambda$. If such an index does not exist, we set $\ell_N = N$.
Observe that if $j \leq r_i$, then, using the node-triangle inequality, $\omega(i) + d(i,j) + \omega(j) \leq \omega(i) + d(i,j) + d(j,r_i) + \omega(r_i) = \omega(i) + d(i,r_i) + \omega(r_i) \leq \lambda$. Therefore,
\begin{equation}\label{eq:r_i}
\omega(i) + d(i,j) + \omega(j) \leq \lambda \text{ iff }  j \leq r_{i}.
\end{equation}
Similarly, if $\ell_N \leq i$, then, using the node-triangle inequality, $\omega(i) + d(i,N) \leq \omega(\ell_N) + d(\ell_N, i) + d(i,N) = \omega(\ell_N) + d(\ell_N, N) \leq \lambda$. Therefore,
\begin{equation}\label{eq:ell_j}
\omega(i) + d(i,N) \leq \lambda \text{ iff } \ell_{N} \leq i.
\end{equation}
The algorithm computes all the indices $r_{i}$, with $1 \leq i < N$, and the index $\ell_N$. Since $\omega(i)\geq \omega(i+1)-d(i,i+1)$, we have that $r_{i} \leq r_{i+1}$. Therefore, all the $r_i$'s can be computed in $O(N)$ time by scanning all the vertices of $P$ in order from $1$ to $N$. Clearly, also $\ell_N$ can be computed in $O(N)$ time by scanning all the vertices of $P$ in order from $N$ downto $1$. As $d(1,N)>\lambda$, we have that $r_1 < N$ and $\ell_N > 1$. We define the following two functions
$$
\bar S(i,j) := d(1,i)+c(i,j)+d(r_{1}+1,j\big)+\omega(r_{1}+1)
$$
and
$$
\bar E(i,j) := d(j,N)+c(i,j)+d(i, \ell_{N}-1)+\omega(\ell_{N}-1).
$$
Observe that both $\bar S(i,j)$ and $\bar E(i,j)$ can be computed in constant time. Moreover, using the graph-triangle inequality, we have that 
\begin{itemize}
\item if $r_1 < j$, then $\bar S(i,j) \leq \bar S(i,j+1)$;
\item if $i < \ell_N$, then $\bar E(i,j) \leq \bar E(i+1,j)$.
\end{itemize}
As the following lemma shows, the values $\bar S(i,j)$ and $\bar E(i,j)$ can be used to understand whether $S(i,j)\leq \lambda$ and $E(i,j)\leq \lambda$, respectively.
\begin{lemma}\label{lemma:new_functions_to_evaluate_stubs}
If $U(i,j) \leq \lambda$, then:
\begin{itemize}
\item $S(i, j) \leq \lambda$ iff $i \leq r_1$ and $\bar S(i, j) \leq \lambda$;
\item $E(i, j) \leq \lambda$ iff $\ell_N \leq j$ and $\bar E(i, j) \leq \lambda$.
\end{itemize}
\end{lemma}
\begin{proof}
We prove the claim only for  $S(i,j)$ as the proof for $E(i,j)$ is similar. 
If $r_1 < i$, then $S(i,j) \geq d(1,i)+\omega(i) \geq d(1,i)+\omega(r_1+1,i)-d(r_1+1,i) = d(1,r_1+1)+\omega(r_1+1) > \lambda$. Therefore, if $S(i,j) \leq \lambda$, then $i \leq r_1$. We break the proof into two cases.

The first case occurs when $j \leq r_1 < N$. $S(i,j)\leq \lambda$ from~(\ref{eq:r_i}), while $\bar S(i,j) = d(1,i)+c(i,j)+d(r_{1}+1,j\big)+ d(r_1+1,N)-d(r_1+1,N)+\omega(r_{1}+1) \leq d(1,i)+c(i,j)+d(r_{1}+1,j\big)+d(r_1+1,N) = U(i,j) \leq \lambda$, and the claim follows.

The second case occurs when $i \leq r_1 < j$. As a consequence $r_1+1 \leq j$.
and  $\bar S(i,j) \leq S(i,j)$ by definition. Therefore, if $S(i,j) \leq \lambda$, so is $\bar S(i,j)$.
If $\bar S(i,j) \leq \lambda$, then let $h$, with $i \leq h \leq j$, be the index such that $S(i,j) = \omega(h)+ \min\big\{d(1, h), d(1,i) + c(i,j) + d(h,j)\big\}$. If $h \leq r_1$, then~(\ref{eq:r_i}) implies that $S(i,j)\leq \lambda$. If $r_1 < h < j$, then
\begin{align*}
S(i,j) 	& = d(1,i) + c(i,j) + d(h,j) + \omega(h) \\
		& = d(1,i) + c(i,j) + d(r_{1}+1,j) + \omega(h) - d(h,r_{1}+1)\\
		& \leq d(1,i) + c(i,j) + d(r_{1}+1,j) + \omega(r_{1}+1) = \bar S(i,j) \leq \lambda.
\end{align*}
The claim follows.
\end{proof}

Let $i < \sigma_i \leq N$ be the maximum index, or $i$ if such an index does not exist, such that $\bar S(i,\sigma_i) \leq \lambda$. Analogously, let $1 \leq \theta_j < j$ be the minimum index, or $j$ if such an index does not exist, such that $\bar E(\theta_j,j) \leq \lambda$. Our algorithm computes all the indices $\sigma_i$, with $1 \leq i < N$, and the indices $\theta_j$, with $1 < j \leq N$. By the graph-triangle inequality, $\bar S(i,j) \leq \bar S(i+1,j)$ as well as $\bar E(i,j)\leq \bar E(i,j-1)$. As a consequence, $\sigma_{i+1} \leq \sigma_i$ and $\theta_{j-1} \geq \theta_{j}$. Therefore, all the $\sigma_i$'s can be computed in $O(N)$ time by scanning all the vertices of $P$ in order from $1$ to $N$. Similarly, all the $\theta_j$'s can be computed in $O(N)$ time by scanning all the vertices of $P$ in order from $N$ downto $1$. The following lemma holds.
\begin{lemma}\label{lemma:new_functions_to_evaluate_stubs_bis}
If $U(i,j) \leq \lambda$, then:
\begin{itemize}
\item $S(i,j) \leq \lambda$ iff $i \leq r_1$ and $j \leq \sigma_i$ (see also Figure \ref{fig:search_problem_example}); 
\item $E(i,j) \leq \lambda$ iff $\ell_N \leq j$ and $\theta_j \leq i$ (see also Figure \ref{fig:search_problem_example}).
\end{itemize}
\end{lemma}
\begin{proof}
We prove the claim only for  $S(i,j)$ as the proof for $E(i,j)$ is similar. Using Lemma~\ref{lemma:new_functions_to_evaluate_stubs}, we have that $S(i, j) \leq \lambda$ iff $i \leq r_1$  and $\bar S(i, j) \leq \lambda$. Therefore, to prove the claim, it is enough to show that $\bar S(i, j) \leq \lambda$ iff $j \leq \sigma_i$. If $j \leq r_1 < N$, then  $\bar S(i,j) = d(1,i)+c(i,j)+d(r_{1}+1,j\big)+ d(r_1+1,N)-d(r_1+1,N)+\omega(r_{1}+1) \leq d(1,i)+c(i,j)+d(r_{1}+1,j\big)+d(r_1+1,N) = U(i,j) \leq \lambda$, and therefore $j \leq \sigma_i$.
Otherwise, if $r_1 < j$, then by definition of $\sigma_i$ and because $\bar S(i,j) \leq \bar S(i,j+1)$, we have that $\bar S(i,j)\leq \lambda$ iff $j \leq \sigma_i$. The claim follows.
\end{proof}

Let $\rho_i$ be the minimum index, or $\perp$ if such an index does not exist, such that $\mu_i \leq \rho_i \leq \sigma_i$ and $i \geq \theta_{\rho_i}$. The algorithm computes $\rho_i$, for every $i=1,\dots,N$.
Since $\mu_i \leq \mu_{i+1}$, $\sigma_{i+1} \leq \sigma_i$, and $\theta_{j-1} \geq \theta_j$, all the indices $\rho_i$ can be computed in $O(N)$ time. The following lemma holds.
\begin{lemma}\label{lemma:property_of_C_function}
Let $\langle P,\delta, \omega, c\rangle$ be an instance of $\wdoap$, where $c$ is a graph-metric and $\omega$ satisfies the node-triangle inequality, and let $\lambda > 0$. 
There exists an index $1 \leq i < N$, such that $\rho_i$ is defined and $C(i,\rho_i) \leq \lambda$ iff the search version of $\wdoap$ on input instance $\langle P, \delta, \omega, c, \lambda \rangle$ admits a feasible solution.
\end{lemma}
\begin{proof}
$(\Rightarrow)$ If $\rho_i$ is defined, then, from Lemma~\ref{lemma:u_function_evaluation} and Lemma~\ref{lemma:new_functions_to_evaluate_stubs_bis}, $U(i,\rho_i),S(i,\rho_i),E(i,\rho_i) \leq \lambda$. Therefore, if $C(i,\rho_i) \leq \lambda$, then $D(i,\rho_i) \leq \lambda$.

\noindent $(\Leftarrow)$ Let $i$ and $j$, with $1 \leq i < j \leq N$, be two indices such that $D(i,j) \leq \lambda$.
Using Lemma~\ref{lemma:u_function_evaluation} and Lemma~\ref{lemma:new_functions_to_evaluate_stubs_bis}, we have that $\mu_i \leq j \leq \sigma_i$ and $i \geq \theta_j$. Since $j$ is a canditate index for the choice of $\rho_i$, $\rho_i \leq j$ by definition. Furthermore, using Lemma~\ref{lemma:u_function_evaluation} and Lemma~\ref{lemma:new_functions_to_evaluate_stubs_bis}, we have that $U(i,\rho_i),S(i,\rho_i),E(i,\rho_i) \leq \lambda$. Finally, as $C(i,j) \leq C(i,j+1)$, $C(i,\rho_i) \leq C(i,j) \leq D(i,j) \leq \lambda$. Therefore $D(i,\rho_i) \leq \lambda$. The claim follows.
\end{proof}
In the following we show how to check whether $C(i,\rho_i) \leq \lambda$ in constant time after an $O(N)$ time preprocessing.
For every $1 \leq i < N$ such that $r_{i} < N$, the algorithm computes 
$$
\Delta(i) = \lambda -\omega(i) + d(i,r_{i}+1) -\omega(r_{i}+1).
$$
Moreover, the algorithm computes $\Delta_{\min} = \min_{1 \leq i < N} \Delta(i)$. Finally, for every $i=1,\dots,N$  for which $\rho_i$ is defined, our algorithm checks whether $d(i,\rho_i)+c(i,\rho_i) \leq \Delta_{\min}$. If there exists $i$ such that $d(i,\rho_i)+c(i,\rho_i) \leq \Delta_{\min}$, then our algorithm returns $(i,\rho_i)$. If this is not the case, then our algorithm returns $\perp$ . The following lemma proves the correctness of our algorithm.
\begin{lemma}\label{lemma:decidability_lemma}
Let $\langle P,\delta, \omega, c\rangle$ be an instance of $\wdoap$, where $c$ is a graph-metric and $\omega$ satisfies the node-triangle inequality, and let $\lambda > 0$. 
The search version of $\wdoap$ on input instance $\langle P, \delta, \omega, c, \lambda \rangle$ admits a feasible solution iff there exists an index $1 \leq i < N$, such that $\rho_i$ is defined and $d(i, \rho_i) + c(i, \rho_i) \leq \Delta_{\min}$.
\end{lemma}
\begin{proof}
($\Rightarrow$) Let $i$ and $j$ be two indices, with $1 \leq i < j \leq N$, such that $D(i,j) \leq \lambda$. Let $1 \leq i_{\min} < N$ be the index such that $\Delta(i_{\min}) = \Delta_{\min}$ and let $r_{\min}=r_{i_{\min}}+1$. Since $\rho_i$ is defined and $C(i,\rho_i) \leq  \lambda$, using Lemma~\ref{lemma:property_of_C_function}, it follows that $\omega(i_{\min})+d(i_{\min},i)+c(i,\rho_i)+d(r_{\min},\rho_i)+\omega(r_{\min}) \leq \lambda$. Therefore,
\begin{align*}
d(i, \rho_i) & + c(i, \rho_i) \leq d(i_{\min},i)  + d(i_{\min},r_{\min}) + d(r_{\min},\rho_i) + c(i,\rho_i)\\
								&  = \Delta(i_{\min})-\lambda+\omega(i_{\min})+\omega(r_{\min})+d(i_{\min},i)+c(i,\rho_i)+d(r_{\min},\rho_i)\\				
								&  \leq \Delta_{\min}-\lambda + \lambda = \Delta_{\min}.
\end{align*}
($\Leftarrow$) We show that for every two indices $k$ and $h$, with $i < k < h < \rho_i$,  $$
\omega(k)+\min\big\{d(k,h),d(1,k)+c(i,\rho_i)+d(h,\rho_i)\big\}+\omega(h) \leq \lambda.
$$
Clearly, if $h \leq r_{k}$, then, from (\ref{eq:r_i}), $\omega(k) + d(k,h) + \omega(h) \leq \lambda$. Therefore, we assume that $h > r_k$. We have that
\begin{align*}
\omega(k)	& + d(i,k)+c(i,\rho_i)+d(h,\rho_i) + \omega(h)\\
	 		& = \omega(k) + d(i,k)+c(i,\rho_i)+d(h,\rho_i) + d(r_{k}+1,h) + \omega(h) - d(r_{k}+1,h)\\
 			& \leq \omega(k) + d(i,k) + c(i,\rho_i) + d(h,\rho_i) + d(r_{k}+1,h) + \omega(r_{k}+1)\\
 			& = \omega(k) + d(i,\rho_i) + c(i,\rho_i) - d(k,r_{k}+1) + \omega(r_{k}+1)\\
	 		& = \omega(k) + d(i,\rho_i) + c(i,\rho_i) - \Delta(k) - \omega(k) - \omega(r_{k}+1) + \lambda + \omega(r_{k}+1)\\
 			& \leq \Delta_{\min} - \Delta(k) + \lambda \leq \lambda. 
\end{align*}
As a consequence, $C(i,\rho_i) \leq \lambda$. The claim now follows from Lemma~\ref{lemma:property_of_C_function}.
\end{proof}
We can conclude this section with the following theorem.
\begin{theorem}
Let $\langle P,\delta, \omega, c\rangle$ be an instance of $\wdoap$, where $c$ is a graph-metric and $\omega$ satisfies the node-triangle inequality, and let $\lambda > 0$. 
The search version of $\wdoap$ on input instance $\langle P, \delta, \omega, c, \lambda\rangle$ can be solved in $O(N)$ time and space.
\end{theorem}
\begin{proof}
Correctness follows from Lemma~\ref{lemma:decidability_lemma}. Since Algorithm~\ref{algorithm:search_problem} computes the values $d(1,N)$ and $\Delta_{\min}$, as well as all the indices $\rho_i$, with $1 \leq i < N$, in $O(N)$ time and space, the claim follows.
\end{proof}

\section{The algorithm for \wdoap}\label{section:metric_trees}

In this section we design an efficient $O(N \log N)$ time and $O(N)$ space algorithm that finds an optimal solution for instances $\langle P, \delta, \omega, c \rangle$ of $\wdoap$, where $c$ is a graph-metric and $\omega$ satisfies the node-triangle inequality. In the rest of the paper we denote by $D^*$ the diameter of an optimal solution to the problem instance and, of course, we assume that $D^*$ is not known by the algorithm. For the rest of this section, unless stated otherwise, $i$ and $j$ are two fixed indices such that $1 \leq i < j \leq N$. Similarly to the notation already used in the previous section, we define $r_i$ as the maximum index such that $i < r_i \leq N$ and $\omega(i) + d(i,r_i)+\omega(r_i) \leq D^*$. If such an index does not exist, then $r_i=i$. Analogously, we define $\ell_N$ as the minimum index such that $1 \leq \ell_N < N$ and $\omega(\ell_N) + d(\ell_N,N) \leq D^*$. If such an index does not exist, then $\ell_N = N$.
Our algorithm consists of the following three phases:
\begin{enumerate}
\item a precomputation phase in which the algorithm computes all the indices $r_{i}$, with $1 \leq i < N$,  and the index $\ell_N$;
\item a search-space reduction phase in which the algorithm reduces the size of the solution search space from $\binom{N}{2}$ to $N-1$ candidates;
\item an optimal-solution selection phase in which the algorithm builds a data structure that is used to evaluate the leftover $N-1$ solutions in $O(\log N)$ time per solution.
\end{enumerate}
Each of the three phases requires $O(N \log N)$ time and $O(N)$ space; furthermore, they all make use of the linear time algorithm for the search version of $\wdoap$ black-box. In the following we assume that $d(1,N) > D^*$, as otherwise, any shorcut returned by our algorithm would be an optimal solution.

\subsection{The precomputation phase}


We perform a binary search over the indices from $1$ to $N$ and use the linear time algorithm for the search version of $\wdoap$ to compute the maximum index $\ell_N$ in $O(N \log N)$ time and $O(N)$ space. Indeed, when our binary search considers the index $k$ as a possible choice of $\ell_N$, it is enough to call the linear time algorithm for the search version of $\wdoap$ with parameter $\lambda=\omega(k)+d(k,N)$ and see whether the algorithm returns either a feasible solution or $\perp$. Due to the node-triangle inequality, in the former case we know that $\ell_N \leq k$, while in the latter case we know that $\ell_N > k$. 

Now we describe how to compute all the indices $r_i$. Because of the node-triangle inequality $r_{i} < N$ iff $i < \ell_N$. Therefore, we only have to describe how to compute $r_{i}$ for every $i < \ell_N$, since if $i \geq \ell_N$, then $r_i=N$. We use the linear time algorithm for the search version of $\wdoap$ and perform a binary search over the set of sorted values $\big\{\omega(i) + d(i,i+1)+\omega(i+1) \mid 1 \leq i < \ell_N\big\}$ to compute the largest value of the set that is less than or equal to $D^*$, if any. This allows us to compute, in $O(N \log N)$ time and $O(N)$ space, the set of all indices $i < \ell_N$ for which $r_i=i$. 
Now, for every index $i < \ell_N$ for which $r_i > i$, we set $a_i=i+1$ and $b_i=N$. Observe that $a_i \leq r_i \leq b_i$. Next, using a two-round binary search, we restrict all the intervals $[a_i, b_i]$'s by updating both $a_i$ and $b_i$ while maintaining the invariant property $a_i \leq r_i \leq b_i$ at the same time. 

Let $X$ be the set of indices $i$, with $1 \leq i < \ell_N$, for which $b_i \geq a_i + 2$. The first round of the binary search ends exactly when $X$ becomes empty. Each iteration of the first round works as follows. For every $i \in X$, the algorithm computes the median index $m_i= \big\lfloor (a_i+b_i)/{2} \big\rfloor$. Next the algorithm computes the {\em weighted} median of the $m_i$'s, say $m_\tau$, where the weight of $m_i$ is equal to $b_i-a_i$. Let
$$
X^+_\tau=\big\{i \in X \mid \omega(i)+d(i,m_i)+\omega(m_i) \geq \omega(\tau)+d(\tau,m_{\tau})+\omega(m_{\tau})\big\}
$$
and
$$
X^-_\tau=\big\{i \in X \mid \omega(i)+d(i,m_i)+\omega(m_i) \leq \omega(\tau)+d(\tau,m_{\tau})+\omega(m_{\tau})\big\}.
$$
Observe that $X=X^+_\tau \cup X^-_\tau$ and $\tau \in X^+, X^-$. 

Now we call the linear time algorithm for the search version of $\wdoap$ with parameter $\lambda = \omega(\tau)+d(\tau, m_\tau)+\omega(m_{\tau})$. If the algorithm finds two indices such that $D(i,j) \leq \lambda$, then we know that $D^* \leq \lambda$ and therefore, for every $i \in X^+_\tau$, we update $b_i$ by setting it equal to $m_i$. If the algorithm outputs $\perp$, then we know that $D^* > \lambda$ and therefore, for every $i \in X^-_\tau$, we update  $a_i$ by setting it equal to $m_i$. We observe that in either case, the invariant property $a_i \leq r_i \leq b_i$ is kept because of (\ref{eq:r_i}). An iteration of the first round of the binary search ends right after the removal of all the indices $i$ such that $b_i=a_i+1$ from $X$. Notice that, in the worst case, the overall sum of the intervals widths at the end of a single iteration is (almost) $3/4$ times the same value computed at the beginning of the iteration. This implies that the first round of the binary search ends after $O(\log N)$ iterations. Furthermore, the time complexity of each iteration is $O(N)$. Therefore, the overall time needed for the first round of the binary search is $O(N \log N)$.

The second round of the binary search works as follows. Because $a_i \leq r_i \leq b_i$ and $b_i \leq a_i+1$ for every $i < \ell_N$ such that $i < r_i$, we have that $r_{i}$ is equal to either $a_i$ or $b_i$. To understand whether either $r_i=a_i$ or $r_i=b_i$, we sort the (at most) $2N$ values 
$$
\Upsilon = \bigcup_{i < \ell_N,\, i < r_i} \big\{\omega(i)+d(i,a_i)+\omega(a_i),\omega(i)+d(i,b_i)+\omega(b_i)\big\}
$$
and use binary search, together with the linear time algorithm for the search version of $\wdoap$, to compute the two consecutive distinct values $D^+,D^- \in \Upsilon$ such that $D^- < D^* \leq D^+$ (if $D^-$ does not exist, then we assume it to be equal to $0$). Finally, we use the two values $D^+$ and $D^-$ to select either $a_i$ or $b_i$. More precisely, if $a_i=b_i$, then $r_i=a_i$. If $a_i \neq b_i$, then by the choice of $D^-$ and $D^+$, either $D^- < \omega(i)+d(i,a_i)+\omega(a_i) \leq D^+$ (i.e., $r_i=a_i$) or $D^- < \omega(i)+d(i,b_i)+\omega(b_i) \leq D^+$ (i.e., $r_i=b_i$). 
The time and space complexities of the second round of the binary search are $O(N \log N)$ and $O(N)$, respectively. We have proved the following lemma.

\begin{lemma}\label{lemma:precomputation_phase}
The precomputation phase takes $O(N \log N)$ time and $O(N)$ space.
\end{lemma}

\subsection{The search-space reduction phase}

In the search-space reduction phase the algorithm computes a set of $N-1$ candidates as optimal shortcut in $O(N \log N)$ time and $O(N)$ space. Let $f(i,j) := \max\big\{U(i,j),\bar E(i,j)\big\}$. Since both $U(i,j)$ and $\bar E(i,j)$ are monotonically non-increasing w.r.t. $j$,\footnote{Observe that $E(i,j)=\max\{\bar E(i,j), \omega(\ell_N)+d(\ell_N,N)\}$ because of the node-triangle inequality. However, since we know that $\omega(\ell_N)+d(\ell_N,N) \leq D^*$ by definition, we can know whether $E(i,j) \leq D^*$ by simply evaluating $\bar E(i,j)$.} $f(i,j)$ is monotonically non-increasing w.r.t. $j$. For every $1 \leq i < N$, our algorithm computes the minimum index $1 < \psi_i \leq N$, if any, such that $f\big(i,\psi_i\big) \leq D^*$. As both $S(i,j)$ and $C(i,j)$ are monotonically non-decreasing w.r.t. $j$, it follows that the set $\big\{\big(i,\psi_i\big) \mid 1 \leq i < N\big\}$ contains an optimal solution to the problem instance.

We compute all the indices $\psi_i$'s using a two-round binary search techique similar to the one we already used in the precomputation phase. First, we set $a_i=i+1$ and $b_i=N$, for every $1 \leq i < N$. Observe that $a_i \leq \psi_i \leq b_i$.
In the two-round binary search, we restrict all the intervals $[a_i, b_i]$'s by updating both $a_i$ and $b_i$ while maintaining the invariant property $a_i \leq \psi_i \leq b_i$ at the same time. 

Let $X$ be the set of indices $i$ for which $b_i \geq a_i + 2$. The first round of the binary search ends exactly when $X$ becomes empty. Each iteration of the first round works as follows. For every $i \in X$, the algorithm computes the median index $m_i= \big\lfloor (a_i+b_i)/{2} \big\rfloor$. Next the algorithm computes the {\em weighted} median of the $m_i$'s, say $m_\tau$, where the weight of $m_i$ is equal to $b_i-a_i$.
Let
$$
X^+_\tau=\big\{i \in X \mid f(i,m_i) \geq f(\tau,m_{\tau})\big\}
\,\,\,\,\,\,\,\,\,\,\text{ and }\,\,\,\,\,\,\,\,\,\,
X^-_\tau=\big\{i \in X \mid f(i,m_i) \leq f(\tau,m_{\tau})\big\}.
$$
Observe that $X=X^+_\tau \cup X^-_\tau$; moreover, $\tau \in X^+, X^-$. 

Now we call the linear time algorithm for the search version of $\wdoap$ with parameter $\lambda = f(\tau,m_{\tau})$. If the algorithm finds two indices such that $D(i,j) \leq \lambda$, then we know that $D^* \leq f(\tau,m_{\tau})$ and therefore, since $f(i,j) \leq f(i,j+1)$, for every $i \in X^+_\tau$, we update $b_i$ by setting it equal to $m_i$. If the algorithm outputs $\perp$, then we know that $D^* > f(\tau,m_{\tau})$ and therefore, by monotonicity of $f$, for every $i \in X^-_\tau$, we update  $a_i$ by setting it equal to $m_i$. We observe that in either case, the invariant property $a_i \leq \psi_i \leq b_i$ is maintained. An iteration of the first round of the binary search ends right after the removal of all the indices $i$ such that $b_i=a_i+1$ from $X$. Notice that, in the worst case, the overall sum of the intervals widths at the end of a single iteration is (almost) $3/4$ times the same value computed at the beginning of the iteration. This implies that the first round of the binary search ends after $O(\log N)$ iterations.  Furthermore, both the time and space complexities of each iteration is $O(N)$. Therefore, the overall time needed for the first round of the binary search is $O(N \log N)$.

The second round of the binary search works as follows. Because $a_i \leq \psi_i \leq b_i$ and $b_i \leq a_i+1$, $\psi_i$ is either equal to $a_i$ or to $b_i$. To understand whether either $\psi_i=a_i$ or $\psi_i=b_i$, we sort the (at most) $2N$ values 
$
\Upsilon=\bigcup_{1 \leq i < N} \big\{f(i,a_i),f(i,b_i)\big\}
$
and use binary search, together with the linear time algorithm for the search version of $\wdoap$, to compute the two consecutive distinct values $D^+, D^- \in \Upsilon$ such that $D^- < D^* \leq D^+$ (if $D^-$ does not exist, then we assume it to be equal to $0$). Finally, we use the two values $D^+$ and $D^-$ to select either $a_i$ or $b_i$. More precisely, if $a_i=b_i$, then $\psi_i=a_i$. If $a_i \neq b_i$, then by the choice of $D^-$ and $D^+$, either $D^- < f(i,a_i) \leq D^+$ (i.e., $\psi_i=a_i$) or $D^- < f(i,b_i) \leq D^+$ (i.e., $\psi_i=b_i$). The time and space complexities of the second round of the binary search are $O(N \log N)$ and $O(N)$, respectively. 
We have proved the following lemma.

\begin{lemma}\label{lemma:candidate_solutions}
The search-space reduction phase takes $O(N \log N)$ time and $O(N)$ space. Furthermore, there exists a shortcut $(i^*,\psi_{i^*})$ such that $D\big(i^*,\psi_{i^*}\big) = D^*$.
\end{lemma}

\subsection{The optimal-solution selection phase}

In the optimal-solution selection phase, we build a data structure in $O(N \log N)$ time and use it to evaluate the quality of the $N-1$ candidates $(1,\psi_1),\dots,(N-1,\psi_{N-1})$ in $O(\log N)$ time per candidate. For every $k=1,\dots,N$, we define 
$$
\phi_k(x):=\omega(k) + \max\big\{d(k,r_{k}) + \omega(r_{k}), x - d(k,r_{k}+1)+\omega(r_{k}+1)\big\}.
$$
Let ${\cal U}(x):=\max\big\{\phi_k(x) \mid 1 \leq k < \ell_N\big\}$ be the {\em upper envelope} of all the functions $\phi_k(x)$.
Observe that each $\phi_k(x)$ is itself the upper envelope of two linear functions. Therefore, ${\cal U}(x)$ is the upper envelope of at most $2N-2$ linear functions. In~\cite{DBLP:journals/ipl/Hershberger89} it is shown how to compute the upper envelope of a set of $O(N)$ linear functions in $O(N \log N)$ time and $O(N)$ space. In the same paper it is also shown how the value ${\cal U}(x)$ can be computed in $O(\log N)$ time, for any $x \in \mathbb{R}$. 

We denote by $x_i = d(i, \psi_i)+c(i, \psi_i)$ the overall weight of the edges of the unique cycle in $P+(i,\psi_i)$. For every $1 \leq i < N$, we compute the value
\begin{equation}\label{eq:last_formula_to_evaluate}
\eta_i=\max\Big\{U\big(i,\psi_i\big), S\big(i,\psi_i\big), E\big(i,\psi_i\big), {\cal U}(x_i)\Big\}.
\end{equation}
The algorithm computes the index $\alpha$ that minimizes $\eta_\alpha$ and returns the shortcut $(\alpha,\psi_\alpha)$ . 

\begin{lemma}\label{lemma:final_lemma}
For every $i$, with $1 \leq i < N$, $C(i,\psi_{i}) \leq {\cal U}(x_{i})$.
\end{lemma}
\begin{proof}
Let $k$ be any index such that $1 \leq k < \ell_N$. We prove the claim by showing that $C(i,\psi_{i}) \leq \phi_k(x_{i})$.  This is shown by proving that for every $h$, with $k < h < \psi_{i}$, $\omega(k)+d_{i,\psi_i}(k,h)+\omega(h) \leq \phi_k(x_{i})$. If $h \leq r_k$, then, using $\omega(h) \leq \omega(r_k)+d(h,r_k)$, we have that
$$
\omega(k)+d_{i,\psi_i}(k,h)+\omega(h) = \omega(k)+d(k,h)+\omega(h) \leq \omega(k) + d(k,r_k)+\omega(r_k) \leq \phi_k(x_{i}),
$$
If $r_k < h$, then, using $\omega(h) \leq \omega(r_k+1)+d(h,r_k+1)$, we have that
\begin{align*}
\omega(k)	& +d_{i,\psi_i}(k,h)+\omega(h) = \omega(k) +d(k,i)+c(i,\psi_{i})+d(\psi_{i},h)+\omega(h)\\
			& \leq \omega(k) + d(k,i)+c(i,\psi_{i})+d(\psi_{i},h)+d(h,r_k+1)+\omega(r_k+1)\\
			& \leq \omega(k) + x_{i}-d(k,r_k+1)+\omega(r_k+1)=\phi_k(x_{i}).
\end{align*}
This completes the proof.
\end{proof}
Let $i^*$ be the index such that $D(i^*,\psi_{i^*}) = D^*$, whose existence is guaranteed by Lemma~\ref{lemma:candidate_solutions}. The following lemma holds.
\begin{lemma}\label{lemma:final_lemma_bis}
${\cal U}(x_{i^*}) \leq D^*$.
\end{lemma}
\begin{proof}
Let $k$ be any index such that $1 \leq k < \ell_N$. We show that $\phi_k(x_{i^*}) \leq D^*$. Clearly, if $\phi_k(x_{i^*}) = \omega(k) + d(k, r_k) + \omega(r_k)$, then $\phi_k(x_{i^*}) \leq D^*$ by definition of $r_{k}$ and the claim would trivially hold.  Therefore, we assume that $\phi_k(x_{i^*}) = \omega(k) + x_{i^*} - d(k,r_{k}+1) + \omega(r_{k}+1)$. 
Next, observe that $k < \psi_{i^*}$ as otherwise
$\omega(k) + d_{i^*,\psi_{i^*}}(k, r_k+1) + \omega(r_k+1) = \omega(k) + d(k, r_k+1) + \omega(r_k+1) > D^*$, thus implying that $D(i^*,\psi_{i^*}) > D^*$.
As a consequence,
\begin{align*}
\phi_k(x_{i^*}) & = \omega(k) + x_{i^*} - d(k,r_{k}+1) + \omega(r_{k}+1) \\
				& = \omega(k) + d(i^*,\psi_{i^*}) + c(i^*,\psi_{i^*}) - d(k,r_{k}+1) + \omega(r_{k}+1) \\
				& \leq \omega(k) + c(i^*,\psi_{i^*}) + d(\psi_{i^*}, r_{k}+1) +  \omega(r_{k}+1)\\
				& \leq \omega(k) + d(k, i^*) + c(i^*,\psi_{i^*}) + d(r_{k}+1,\psi_{i^*}) + \omega(r_{k}+1)\\
				& = \omega(k) + d_{i^*,\psi_{i^*}}(k,r_{k}+1) + \omega(r_{k}+1) \leq D^*.
\end{align*}
The claim follows.
\end{proof}
We can finally conclude this section by stating the main results of this paper.
\begin{theorem}\label{theorem:tree-metric_node_weighted_path}
$\wdoap$ can be solved in $O(N \log N)$ time and $O(N)$ space.
\end{theorem}
\begin{proof}
From Lemma~\ref{lemma:final_lemma}, we have that $\eta_i \geq D(i,\psi_i)$. However, from Lemma~\ref{lemma:final_lemma_bis}, $\eta_{i^*} \leq D^*$. Therefore, the index $\alpha$ computed by the algorithm satisfies $D(\alpha,\psi_\alpha) \leq \eta_\alpha \leq \eta_{i^*} \leq D^*$, i.e., the shortcut $(\alpha,\psi_\alpha)$ returned by the algorithm is an optimal solution.

Concerning the time and space complexities of the algorith, in Lemma~\ref{lemma:precomputation_phase} we proved that the precomputation phase takes $O(N \log N)$ time and $O(N)$ space, in Lemma~\ref{lemma:candidate_solutions} we proved that the search-space reduction phase takes $O(N \log N)$ time and $O(N)$ space, and, in Lemma~\ref{lemma:time_complexity_real_functions}, we proved that $U\big(i,\psi_i\big)$, $S\big(i,\psi_i\big)$, and $E\big(i,\psi_i\big)$ can be computed in $O(\log N)$ time after a precomputation phase of $O(N)$. Since the value ${\cal U}(x_i)$ can be computed in $O(\log N)$ time~\cite{DBLP:journals/ipl/Hershberger89}, all the $N-1$ values $\eta_i$ can be computed in $O(N \log N)$ time and $O(N)$ space. This completes the proof.
\end{proof}
\begin{theorem}\label{theorem:metric_tree}
$\doap$ on trees embedded in a (graph-)metric space can be solved in $O(n \log n)$ time and $O(n)$ space.
\end{theorem}
\begin{proof}
Let $\langle T, \delta, c \rangle$ be an input instance of $\doap$, where $T$ is a tree and $c$ is a metric function. In Lemma~\ref{lemma:equivalence_to_metric_instance} we proved that, if $P=(v_1,\dots,v_N)$ is a diametral path of $T$, then an optimal solution of the problem is a shortcut whose endvertices are both vertices of $P$. Furthermore, in Lemma~\ref{lemma:Dij} we proved that the problem instance can be reduced in $O(n)$ time to the corresponding induced instance $\langle P, \delta, w, c \rangle$ of $\wdoap$. Therefore, an optimal solution of the induced problem instance of $\wdoap$ is also an optimal solution of the instance of $\doap$. Since the induced instance of $\wdoap$ is solvable in $O(N \log N)=O(n \log n)$ time and $O(N)=O(n)$ space (see Theorem~\ref{theorem:tree-metric_node_weighted_path}), the claim follows.
\end{proof}

\section{The algorithm for general problem instances}\label{section:general_instances}

In this section we design an algorithm that solves $\doap$ in $O(n^2)$ time and space. 
The algorithm first queries the oracle to explicitly compute all the values $c(u,v)$, for every $u,v \in V(T)$, in $O(n^2)$. Next the algorithm computes, in $O(n)$ time, a diametral path $P=(v_1,\dots,v_N)$ of $T$ and all the values $w(v_i)$ and $\omega(v_i)$. Now we show how the graph-metric function $\bar c$ restricted to pair of vertices $v_i$ and $v_j$ of $P$ can be computed in $O(n^2)$ time. For every $1 \leq i < j \leq N$, the algorithm sets 
$$
\hat c(v_i,v_j) = \min\big \{d(u,v_i)+c(u,v)+d(v_j,v) \mid u \in V(T_i),v \in V(T_j)\big\}.
$$
Next, for every $i$ from $1$ up to $N-1$ and for every $j$ from $i+1$ up to $N$, the algorithm computes 
$$
\tilde c(v_i,v_j)=\min\big\{\hat c(v_i,v_j), \tilde c(v_i,v_{j-1})+d(v_{j-1},v_j), \tilde c(v_{i-1},v_j)+d(v_{i-1},v_i)\big\}.
$$
Finally, for every $i$ from $N$ downto $2$ and for every $j$ from $i-1$ downto $1$, the algorithm computes
$$
\bar c(v_i,v_j)=\min\big\{\tilde c(v_i,v_j), \bar c(v_i,v_{j+1})+d(v_{j+1},v_j), \bar c(v_{i+1},v_j)+d(v_{i+1},v_i)\big\}.
$$

\begin{lemma}\label{lemma:computation_graph_metric_closure}
The function $\bar c$ computed by the algorithm corresponds to the graph-metric closure of $c$.
\end{lemma}
\begin{proof}
Let $i$ and $j$ be any two indices such that $1 \leq i < j \leq N$. Let $u$ and $v$ be the two vertices of $T$ such that the graph-metric cost function for the pair $v_i$ and $v_j$ is equal to $d(v_i,u)+c(u,v)+d(v,v_j)$. 

By definition, the value $\hat c(v_i,v_j)$ refers to the length of a (not necessarily simple) path between $v_i$ and $v_j$ in $T$ plus a non-tree edge. A simple proof by induction shows that this is the case also for $\tilde c(v_i,v_j)$ and $\bar c(v_i,v_j)$. Therefore, the value $\bar c(v_i,v_j)$ computed by the algorithm is always greater than or equal to $d(v_i,u)+c(u,v)+d(v,v_j)$.

We conclude the proof by showing that the value $\bar c(v_i,v_j)$ computed by the algorithm is at most $d(v_i,u)+c(u,v)+d(v,v_j)$. Let $T_k$ and $T_h$, with $k \leq h$, be the two trees of $T-E(P)$ that contain $u$ and $v$, respectively. Observe that $\hat c(v_h,v_k)\leq d(v_k,u)+c(u,v)+d(v,v_h)$ by definition. As a consequence, also $\tilde c(v_k,v_h), \bar c(v_k,v_h) \leq d(v_k,u)+c(u,v)+d(v,v_h)$. We divide the proof into the following three cases:
\begin{itemize}
\item $k \leq i \leq h \leq j$;
\item $i \leq k \leq j \leq h$;
\item $i \leq k \leq h \leq j$.
\end{itemize}

We consider the case  $k \leq i \leq h \leq j$. A simple proof by induction shows that $\tilde c(v_i,v_h) \leq \tilde c(v_k,v_h)+d(v_k,v_i)$ before the value $\tilde c(v_i,v_j)$ is computed by the algorithm. Similarly, $\tilde c(v_i,v_j) \leq \tilde c(v_i,v_h)+d(v_j,v_h)$, from which we derive 
$$
\bar c(v_i,v_j) \leq \tilde c(v_i,v_j) \leq d(v_k,v_i)+\tilde c(v_k,v_h)+d(v_h,v_j) \leq d(v_i,u)+c(u,v)+d(v,v_j).
$$

The proof for the case $i \leq k \leq j \leq h$ is similar. A simple proof by induction shows that $\bar c(v_k,v_j) \leq \bar c(v_k,v_h)+d(v_j,v_h)$ before the value $\bar c(v_i,v_j)$ is computed by the algorithm. Similarly, $\bar c(v_i,v_j) \leq \bar c(v_k,v_j) + d(v_i,v_k)$, from which we derive
$$
\bar c(v_i,v_j) \leq d(v_j, v_h) + \bar c(v_k,v_h) + d(v_i,v_k) \leq  d(v_i,u)+c(u,v)+d(v,v_j).
$$

We consider the case $i \leq k \leq h \leq j$. A simple proof by induction shows that $\tilde c(v_k,v_j) \leq \tilde c(v_k,v_h)+d(v_h,v_j)$ before the value $\bar c(v_i,v_j)$ is computed by the algorithm. Similarly, $\bar c(v_i,v_j) \leq \tilde c(v_k,v_j)+d(v_k,v_i)$, from which we derive
$$
\bar c(v_i,v_j) \leq d(v_h,v_j)+ \tilde c(v_k,v_h) + d(v_i,v_k) \leq d(v_i,u)+c(u,v)+d(v,v_j).
$$
This completes the proof.
\end{proof}

\begin{theorem}
$\doap$ can be solved in $O(n^2)$ time and space.
\end{theorem}
\begin{proof}
In Lemma~\ref{lemma:reduction_to_tree_metric_instances} we proved that it is enough to compute an optimal shortcut w.r.t. the graph-metric closure $\bar c$. In Lemma~\ref{lemma:equivalence_to_metric_instance} we proved that an optimal solution of the problem is a shortcut whose endvertices are both vertices of $P$. Since the instance $\langle P, \delta, \omega, \bar c\rangle$ of $\wdoap$ induced by $\langle T, \delta, w, \bar c\rangle$ can be computed in $O(n^2)$ time and space (see Lemma~\ref{lemma:computation_graph_metric_closure}), the claim follows from Theorem~\ref{theorem:metric_tree}.
\end{proof}

\section{The \texorpdfstring{$(1+\varepsilon)$}{1+epsilon}-approximation algorithm for \doap}\label{section:approximation_algorithm}

In this section we design an algorithm that computes a $(1+\varepsilon)$-approximate solution of $\wdoap$ in $O\left(N+\frac{1}{\varepsilon}\log \frac{1}{\varepsilon}\right)$ time and $O(N+1/\varepsilon)$ space. First, we prove an upper bound to $d(1,N)$ as a function of $D^*$.

\begin{lemma}\label{lemma:relation_between_path_and_D_star}
$d(1,N) \leq 3D^*$.
\end{lemma}
\begin{proof}
Let $r_1$ be the maximum index such that $d(1,r_1) \leq D^*$; similarly, let $\ell_N$ be the minimum index such that $d(\ell_N,N) \leq D^*$.  If $\ell_N \leq r_1$, then $ d(1,N) \leq d(1,r_1)+d(\ell_N,N) \leq 2D^*$ and the claim follows. Therefore, we assume that $r_1 < \ell_N$. We show that $d(r_1,\ell_N) \leq D^*$. Let $i^*$ and $j^*$, with $1 \leq i^* < j^* \leq N$, be two indices such that $D(i^*,j^*) = D^*$. Using the triangle inequality $d(r_1,\ell_N) \leq d(i^*,r_1)+c(i^*,j^*)+d(j^*,\ell_N)$ we obtain
\begin{equation}\label{eq:cost_non_edge_i_j}
c(i^*,j^*)\geq d(r_1,\ell_N)-d(i^*,r_1)-d(j^*,\ell_N).
\end{equation}
Furthermore, either $d(r_1,\ell_N) \leq D^*$ or 
\begin{equation}\label{eq:r1_to_lN}
d(i^*,r_1)+c(i^*,j^*)+d(j^*,\ell_N) \leq D^*.
\end{equation}
In the former case the claim $d(r_1,\ell_N) \leq D^*$ trivially holds. In the latter case, by plugging~(\ref{eq:cost_non_edge_i_j}) into~(\ref{eq:r1_to_lN}) we obtain
$d(r_1,\ell_N) \leq d(1,i^*)+c(i^*,j^*)+d(j^*,N) \leq D^*$.
Therefore, $d(r_1,\ell_N) \leq D^*$ in both cases. Hence, $d(1,N) = d(1,r_1)+d(r_1,\ell_N)+d(\ell_N,N) \leq 3D^*$.
\end{proof}

Let $L=\varepsilon\cdot d(1,N)/18$ and observe that $d(1,N)$ can be subdivided into at most $18/\varepsilon$ intervals of width $L$ each. For each $k$, with $1 \leq k \leq 18/\varepsilon$, the algorithm computes the {\em representative} index $\gamma_k$ that maximizes $\omega(\gamma_k)$  among the indices $i$ such that $(k-1)L \leq d(1,i) \leq kL$ (i.e., the indices of the $k$-th interval). Notice that all the values $\omega(\gamma_k)$ can be computed in $O(N)$ time by scanning the vertices of $P$ in order from $1$ to $N$. 
Next, the algorithm finds an optimal solution of the input instance restricted to the set of the representative vertices in $O(\varepsilon^{-1} \log \varepsilon^{-1})$ time and $O(\varepsilon^{-1})$ space using the algorithm we described in the previous section. We can prove the following theorems.
\begin{theorem}\label{theorem:metric_node_weighted_path}
For every $\varepsilon > 0$, we can compute a $(1+\varepsilon)$-approximate solution of $\wdoap$ in $O\left(N+\frac{1}{\varepsilon}\log \frac{1}{\varepsilon}\right)$ time and $O\left(N+\frac{1}{\varepsilon}\right)$ space.
\end{theorem}
\begin{proof}
We already proved that the time and space complexities of the algorithm are $O\left(N+\frac{1}{\varepsilon}\log \frac{1}{\varepsilon}\right)$ and $O\left(N+\frac{1}{\varepsilon}\right)$, respectively. Therefore, we only have to prove that the algorithm returns a $(1+\varepsilon)$-approximate solution.

Let $i^*$ and $j^*$, with $1 \leq i^* < j^* \leq N$, be two indices such that $D(i^*,j^*) = D^*$. We can assume that $(i-1)L \leq d(1,i^*) \leq iL$ and $(j-1)L\leq d(1,j^*)\leq jL$, with $i < j$. Indeed, if $i=j$, then the addition of the shortcut $(i^*,j^*)$ to $P$ cannot shorten paths in $P$ by more than $L$, and therefore, $d(1,N)-L \leq D(i^*,j^*) \leq D^*$, i.e.,
$D(i,j)\leq d(1,N) \leq D^*+\varepsilon D^*/6 \leq (1+\varepsilon)D^*$ for every shortcut $(i,j)$.
Observe that both $\gamma_i$ and $i^*$ belong to the $i$-th interval as well as both $\gamma_j$ and $j^*$ both to the $j$-th interval. As a consequence, using the triangle inequality we have that
\begin{equation}\label{eq:gammas}
c(\gamma_i,\gamma_j) \leq d(i^*,\gamma_i)+c(i^*,j^*)+d(j^*,\gamma_j) \leq 2L + c(i^*,j^*).
\end{equation}
Let $k$ and $h$ be two indices such that $1 \leq k < h \leq N$ and let $d(k,i^*)+c(i^*,j^*)+d(h,j^*)$ be the length of a shortest path between $k$ and $h$ in $P+(i^*,j^*)$ chosen among the set of all paths passing through $(i^*,j^*)$. Using (\ref{eq:gammas}), we can bound the length of the shortest path between $k$ and $h$ in $P+(\gamma_i,\gamma_j)$, chosen among the set of all paths passing through $(\gamma_i,\gamma_k)$,  with
\begin{align*}
d(k,\gamma_i)	& + c(\gamma_i,\gamma_j)+d(\gamma_j,h) \\
				& \leq d(k,i^*)+d(i^*,\gamma_i)+2L + c(i^*,j^*)+d(j^*,\gamma_j)+d(j^*,h)\\
				& \leq d(k,i^*)+c(i^*,j^*)+d(h,j^*)+4L.
\end{align*}
As a consequence, the shortest path between any two vertices in $P+(\gamma_i,\gamma_j)$ is longer than the corresponding shortest path in $P+(i^*,j^*)$ by an additive term of at most $4L$. Therefore, $D(\gamma_i,\gamma_j) \leq D(i^*,j^*) + 4L \leq D^*+4L$. This implies that the value of $D(\gamma_i, \gamma_j)$ measured w.r.t. the path instance restricted to the representative vertices only is at most $D^*+4L$. Let $(\gamma_{k},\gamma_{h})$ be the optimal shortcut of the instance restricted to the representative vertices only returned by the algorithm. When we measure the value $D(\gamma_k,\gamma_h)$ measured w.r.t. all the vertices of $P$, we need to take into account also the additive term paid as the distance between any vertex and its corresponding representative plus the weight of the considered vertex. Since the representative index is the one that maximizes its weight w.r.t. all the indices of the interval it belongs to and since the distance between any vertex of $P$ and its corresponding representative is at most $L$, the value $D(\gamma_k,\gamma_h)$ measured w.r.t. all the vertices of $P$ is at most the value $D(\gamma_k,\gamma_h)$ measured w.r.t. all the representative vertices only plus an additive term of at most $2L$ (at most $L$ for $\gamma_k$ and at most $L$ for $\gamma_h$). Therefore, by choice of $L$ and from Lemma~\ref{lemma:relation_between_path_and_D_star}, the value $D(\gamma_k,\gamma_h)$ measured w.r.t. all the vertices of $P$ is at most $D^*+4L+2L = D^* + \varepsilon\cdot d(1,N)/3 \leq (1+\varepsilon)D^*$.
\end{proof}

\begin{theorem}\label{theorem:_apx_metric_tree}
For every $\varepsilon >0$, we can compute a $(1+\varepsilon)$-approximate solution of $\doap$ in $O\left(n + \frac{1}{\varepsilon}\log\frac{1}{\varepsilon}\right)$ time and $O\left(n+\frac{1}{\varepsilon}\right)$ space when the tree is embedded in a metric space.
\end{theorem}
\begin{proof}
In Lemma~\ref{lemma:equivalence_to_metric_instance} we proved that, if $P=(v_1,\dots,v_N)$ is a diametral path of $T$, then an optimal solution of the problem is a shortcut whose endvertices are both vertices of $P$. Furthermore, in Lemma~\ref{lemma:Dij} we proved that the instance of $\doap$ can be reduced in $O(n)$ time to the corresponding induced instance of $\wdoap$, and the reduction preserves the approximation. Therefore, a $(1+\varepsilon)$ -approximate solution of the induced problem instance is also a $(1+\varepsilon)$-approximate solution of our problem instance. Since a $(1+\varepsilon)$-approximate solution for the induced problem instance can be found in $O\left(N +\frac{1}{\varepsilon}\log \frac{1}{\varepsilon}\right)=O\left(n +\frac{1}{\varepsilon}\log \frac{1}{\varepsilon}\right)$ time and $O\left(N+\frac{1}{\varepsilon}\right)=O\left(n+\frac{1}{\varepsilon}\right)$ space (see Theorem~\ref{theorem:metric_node_weighted_path}), the claim follows.
\end{proof}

\bibliography{bibliography}

\end{document}